\documentclass[10pt]{article}
 \usepackage[letterpaper=true,colorlinks=true,pdfpagemode=none,urlcolor=blue,linkcolor=blue,citecolor=blue,pdfstartview=FitH]{hyperref}
\usepackage{amsmath,amsfonts,enumerate,amssymb}
\usepackage{graphicx}
\usepackage{color}
\usepackage{subfigure}
\usepackage{mathtools}

\newif\ifblog
\newif\iftex
\blogfalse
\textrue

\usepackage{ulem}
\def\em{\it}
\def\emph#1{\textit{#1}}

\def\P{{\mathbb P}}
\def\E{{\mathbb E}}

\def\R{{\mathbb R}}

\def\Q{{\mathbb Q}}
\def\L{{\mathcal L}}

\newtheorem{theorem}{Theorem}
\newtheorem{lemma}[theorem]{Lemma}

\newtheorem{corollary}[theorem]{Corollary}
\newtheorem{proposition}[theorem]{Proposition}
\newtheorem{example}[theorem]{Example}
\newtheorem{remark}[theorem]{Remark}
\newtheorem{assumption}[theorem]{Assumption}
\newenvironment{proof}{\noindent {\sc Proof:}}{$\Box$} %\medskip}

\numberwithin{equation}{section}
\numberwithin{theorem}{section}

\DeclareMathOperator*{\argmin}{arg\,min}

\def\F{{\mathcal F}}
\def\H{{\mathcal H}}
\def\G{{\mathcal G}}

\title{Outperformance Portfolio Optimization via\\ the  Equivalence of Pure and Randomized Hypothesis Testing}
\author{Tim Leung\thanks{Department of Industrial Engineering and Operations Research, Columbia University, New York, NY 10027; email:\,\mbox{leung@ieor.columbia.edu}.} \and
  Qingshuo Song\thanks{Department of Mathematics, City University of Hong Kong, Hong Kong; email:\,\mbox{song.qingshuo@cityu.edu.hk}.}
  \and  Jie Yang\thanks{Department of Mathematics, Statistics, and
    Computer Science, University of Illinois at Chicago,  Chicago, IL
    60607; email: jyang06@math.uic.edu }}
\begin{document}
\date{\today}
\maketitle
\begin{abstract}
  We study the portfolio problem of maximizing  the outperformance
  probability over a random benchmark through dynamic trading with a
  fixed initial capital. Under a general incomplete market framework,
  this stochastic control problem can be formulated as  a composite pure
  hypothesis testing problem. We analyze the connection between this
  \emph{pure} testing problem and its \emph{randomized} counterpart, and
  from latter we derive  a dual representation for the maximal
  outperformance probability. Moreover,
  in a complete market setting, we provide a closed-form solution to the
  problem of beating a leveraged exchange traded fund.  For a general
  benchmark under an incomplete stochastic factor model, we  provide the
  Hamilton-Jacobi-Bellman PDE characterization for the maximal
  outperformance probability.
\end{abstract}
\vspace{10pt}
\begin{small}
  \vspace{3pt}
\end{small}

\newpage
\section{Introduction}Portfolio optimization problems with an
objective  to exceed a given benchmark arise very commonly in
portfolio management among both institutional and individual
investors.  For many hedge funds, mutual funds and other investment
portfolios, their performance is evaluated relative to the market
indices, e.g. the S\&P 500 Index, and Russell 1000 Index.
In this paper, we consider the problem  of maximizing  the
outperformance probability over a random benchmark through a dynamic
trading with a fixed initial capital.    Specifically, given an
initial capital $x>0$ and a random benchmark $F$,  how can one
construct a dynamic trading strategy $(\pi_t)_{0\le t\le T}$ in order
to maximize the probability of the ``success event" where the terminal trading wealth
$X^{x,\pi}_T$ exceeds $F$, i.e.  $\mathbb{P}\{X^{x,\pi}_T \ge
F\}$?

% The clarity of benchmark gives outperformance portfolio optimization an advantage over the standard utility maximization which involves specifying a subjective intangible utility function.

In the existing literature,  outperformance portfolio optimization has been
studied by  \cite{BHS10,Browne99,SC99} among others. It has also
been studied in the context of quantile hedging by F\"{o}llmer and
Leukert \cite{FL99}.  In particular, F\"{o}llmer and Leukert  show
that the quantile hedging problem can be formulated as a \emph{pure}
hypothesis testing problem. In statistical terminology, this approach seeks to determine a \emph{test}, taking values 0 or 1,  that minimizes the probability of type-II-error, while limiting  the probability of type-I-error by  a pre-specified acceptable significance level. The maximal success probability  can be interpreted as the \emph{power} of the  test. The F\"{o}llmer-Leukert  approach   permits the use of an important  result from statistics, namely, the Neyman-Pearson
Lemma (see, for example, \cite{Lehmann}), to characterize the optimal success event and  determine its probability.

On the other hand, the outperformance portfolio optimization  can also
be viewed as a special case of  shortfall risk minimization, that is,
to minimize the quantity
$\rho(-(F-X^{x,\pi}_T)^+)$ for some specific risk measure $\rho(\cdot)$.
As is well known (see \cite{Cvi00,FS02,Rud07,Sch04}), the shortfall
risk minimization with a convex risk measure can be solved via its
equivalent \emph{randomized} hypothesis testing problem. In fact, the
problem to maximize the success probability $\mathbb{P}\{X^{x,\pi}_T \ge F\}$
is  equivalent to minimizing the shortfall risk
$\mathbb{P}\{X^{x,\pi}_T < F\} = \rho(-(F-X^{x,\pi}_T)^+)$ with
respect to the risk measure defined by $\rho(Y) := \mathbb{P}\{Y<0\}$
for any random variable $Y$.
However, this risk measure $\rho(\cdot)$ does not satisfy
either convexity or continuity.
Hence, a natural question is:
\begin{enumerate}
\item [(Q)] Is the outperformance optimization
  problem  equivalent
  to the randomized hypothesis testing?
\end{enumerate}

In Section~\ref{sec:31}, we show that the outperformance portfolio optimization in a general incomplete market is
equivalent to a pure hypothesis testing.  Moreover, we illustrate that the outperformance probability, or equivalently, the associated pure hypothesis testing value, can be strictly smaller than the corresponding randomized hypothesis testing (see Examples~\ref{exm:pr} and \ref{exm:bin}). Therefore, the answer to (Q)
is  negative in general. This also motivates us to analyze the
sufficient conditions for the equivalence of pure  and randomized
hypothesis testing
problems (see Theorem~\ref{thm:np}).
In turn, our result is  applied to give the sufficient conditions for the equivalence of outperformance portfolio optimization  and the corresponding randomized hypothesis testing problem (see Theorem~\ref{thm:qhic}).

The main benefit of such an equivalence is that it allows us to
utilize the representation of the randomized testing value to compute
the optimal outperformance probability.  Moreover, the sufficient
conditions established herein are amenable for the verification
and are
applicable to many typical finance markets. We provide detailed
illustrative examples in Section~\ref{sec:cm} for a complete market
and  Section~\ref{sect:stochvol} for a stochastic volatility model.

Among other results, we provide an explicit solution to the problem of
outperforming a leveraged fund in a complete market. In a stochastic
volatility market, we show that, for a constant or stock benchmark,
the investor may optimally assign a zero volatility risk premium,
which corresponds to the  minimal martingale measure (MMM). This in
turn allows for  explicit solution for the success probability in a
range of cases in this incomplete market. With the general form of
benchmark, the value function can be characterized by HJB equation in
the framework of stochastic control theory.

The paper is structured as follows. In Section \ref{sect-hypo}, we
analyze the generalized composite pure and randomized hypothesis
testing problems, and study their equivalence.  Then, we apply the
results to solve the related outperformance portfolio optimization  in
Section \ref{sect-finance}, with examples in
both complete  and incomplete diffusion
markets. Section~\ref{sec:conclusions}  concludes the paper and discusses a number of extensions.  Finally,
we include a number of examples and proofs in the  Appendix.

\section{Generalized Composite Hypothesis Testing}\label{sect-hypo}
In the background, we fix a complete probability space
$(\Omega, \mathcal{F}, \mathbb{P})$.
Denote by $\E[\, \cdot \, ]$ the
expectation under $\P$, and $L^{0,+}$ the space of all
non-negative $\mathcal{F}$-measurable
random variables, equipped with
the topology endowed by the convergence
in probability.
The  randomized tests and pure tests are represented by the two
collections of random
variables taking values in $[0,1]$ and $\{0,1\}$ respectively, and are
denoted by
\[\mathcal{X} = \{X: \Omega/\mathcal{F} \mapsto
[0,1]/\mathcal{B}([0,1])\} \quad \text{and} \quad \mathcal{I} = \{X:
\Omega/\mathcal{F}
\mapsto \{0,1\}/2^{\{0,1\}}\}.\] In addition, $\mathcal{G}$ and
$\mathcal{H}$ are two given collections of
non-negative $\mathcal{F}$-measurable  random variables.

\subsection{Randomized Composite Hypothesis Testing}\label{sec:21}
First, we  consider  a  randomized composite hypothesis testing problem. For $x>0$, define
\begin{align}&
  V(x): = \sup_{X\in \mathcal{X}}\inf_{G\in \mathcal{G}} \mathbb{E}[ GX]\label{Hypt_1}\\
  \text{ subject to } &\sup_{H\in \mathcal{H}} \mathbb{E} [HX]\le x.\label{Hypt_2}\end{align}
  From the statistical viewpoint, ${\mathcal G}$ and ${\mathcal H}$ correspond to the collections of alternative hypotheses
and null hypotheses, respectively. The solution $X$ can be viewed as the most
powerful test, and $V(x)$ is the power of $X$,  where $x$ is the significance level
or the size of the test.

For any  set of random variables
$\mathcal{\tilde H} \subset L^{0,+}$, we define a collection of randomized tests by \begin{equation}
  \label{eq:eqchi}
  \mathcal{X}^{\mathcal{\tilde H}}_x :=
  \{X\in \mathcal{X}: \mathbb{E}[H X] \le x, \
  \forall H \in \mathcal{\tilde H}\}.
\end{equation}
Then, the problem in  (\ref{Hypt_1})-(\ref{Hypt_2}) can be equivalently  expressed as
\begin{equation}
  \label{eq:rnp}
  V(x) = \sup_{X\in \mathcal{X}^\mathcal{H}_x} \inf_{G\in \mathcal{G}} \mathbb{E} [GX].
\end{equation}
When no ambiguity arises, we will denote $\mathcal{X}_x = \mathcal{X}_x^\mathcal{H}$ for simplicity.

For the upcoming results, we denote
the convex hull of $\mathcal{H}$
by $co(\mathcal{H})$, and the closure
(with respect to the topology endowed by the
convergence in probability)
of  $co(\mathcal{H})$
by $\overline{co(\mathcal{H})}$. Also, we define the set
\begin{equation}
  \label{eq:Hx}
  \mathcal{H}_x := \{H\in L^{0,+}: \mathbb{E}[HX] \le x, \ \forall X
  \in \mathcal{X}_x^\mathcal{H}\}.
\end{equation}
From the definitions together with Fatou's lemma,
it is straightforward to check
that $\mathcal{H}_x$ is  convex and closed, containing
$\mathcal{H}$. Furthermore, we observe that
$\mathcal{X}_x^\mathcal{H} = \mathcal{X}_x^{\mathcal{\tilde H}}$
for an arbitrary $\mathcal{\tilde H}$ satisfying
$\mathcal{H} \subset \mathcal{\tilde H} \subset \mathcal{H}_x$.
Hence, the
randomized testing problem in \eqref{Hypt_1}-\eqref{Hypt_2}, and
therefore, $V(x)$ in \eqref{eq:rnp} will stay invariant if
$\mathcal{H}$ is replaced by $\mathcal{\tilde H}$ as such.
More precisely, we  have

\begin{lemma}
  \label{prop:H}
  Let $\mathcal{\tilde H}$ be an arbitrary set satisfying
  $\mathcal{H}  \subset \mathcal{\tilde H}
  \subset \mathcal{H}_x$.
  Then, $V(x)$ in \eqref{eq:rnp} is equivalent to
  \begin{align}\label{Veqv}V(x) =
    \sup_{X\in \mathcal{X}_x^{\mathcal{\tilde H}}}
    \inf_{G\in \mathcal{G}} \mathbb{E} [GX].\end{align}
  In particular, one can take $\mathcal{\tilde H} ={co(\mathcal{H})}$ or
  $\mathcal{H}_x$.
\end{lemma}

This randomized hypothesis testing problem  is similar to that studied
by Cvitani{\'c} and  Karatzas \cite{CK01}, except that  $G$ and $H$ in
\eqref{Hypt_1}-\eqref{Hypt_2} and \eqref{eq:rnp} are not necessarily
the Radon-Nikodym derivatives for probability measures.
In this slight generalization,  $\mathbb{E}[H]$ can vary among
$\mathcal{H}$, which allows for statistical hypothesis testing  with
different significance
levels depending on $H$. To see this, one can divide \eqref{Hypt_2}  by $\E[H]$ for each $H\in \mathcal{H}$, resulting in a confidence level of $x/\E[H]$ (see also Remark 5.2 in  \cite{KR10}). Similar to \cite{CK01} and \cite{LSYproceedings}, we  make the following standing assumption:
\begin{assumption}
  \label{a:1}
  Assume that  $\mathcal{G}$ and $\mathcal{H}$ are subsets of
  $L^{0,+}$ with
  $\displaystyle \sup_{X\in \mathcal{G} \cup \mathcal{H}} \mathbb{E}[X] <\infty$,
  and $\mathcal{G}$ is convex and closed.
\end{assumption}
The following theorem gives the characterization of the
solution for \eqref{eq:rnp}.
\begin{theorem}\label{thm:rnp} Under Assumption \ref{a:1}, there exists
  $(\hat G, \hat H, \hat a, \hat X) \in\! \mathcal{G} \times
    \overline{co(\mathcal{H})} \times [0,\infty) \times
    \mathcal{X}_x$  satisfying
  \begin{equation}
    \label{eq:1}
    \hat X = I_{\{\hat G > \hat a \hat H\}} +   B I_{\{ \hat G=\hat a \hat H
      \}}, \hbox{ for some } B: \Omega/\mathcal{F} \mapsto
    [0,1]/\mathcal{B}([0,1]),
  \end{equation}
  \begin{equation}
    \label{eq:2}
    \mathbb{E}[ H \hat X] \le \mathbb{E}[\hat H \hat X] = x, \quad
    \forall H \in \mathcal{H},
  \end{equation}
  and
  \begin{equation}
    \label{eq:GX}
    \mathbb{E} [\hat G \hat X] \le \mathbb{E} [ G \hat X], \quad
    \forall G \in \mathcal{G}.
  \end{equation}
  In particular, $\hat X$ and $B$ satisfying \eqref{eq:1}-\eqref{eq:GX}
  can be chosen to be measurable with respect to $\sigma(\mathcal{G}
  \cup \mathcal{H})$, the smallest $\sigma$-algebra generated by
  the random variables in $\mathcal{G}\cup \mathcal{H}$.
  Moreover, $V(x)$ of \eqref{eq:rnp} is given by
  \begin{equation}
    \label{eq:Vr}
    V(x) = \mathbb{E}[\hat G\hat X] =
    \inf_{a \ge 0} \big\{xa + \inf_{\mathcal{G} \times co(\mathcal{H})}
    \mathbb{E} [(G- a H)^+]\big\},
  \end{equation}which is  continuous, concave, and
  non-decreasing in  $x \in [0,\infty)$.
  Furthermore, $(\hat G, \hat H)$ and  $(\hat G, \hat H, \hat a)$
  respectively   attain the infimum of
  \begin{equation}
    \label{eq:gha}
    \mathbb{E} [(G- \hat a H)^+],\quad \hbox{and}\quad xa + \mathbb{E} [(G- a H)^+].
  \end{equation}
\end{theorem}

\begin{proof}
  First, we apply the equivalence between \eqref{eq:rnp} and
  \eqref{Veqv} from Lemma \ref{prop:H}, and the fact that
  $\mathcal{X}^{\mathcal{H}}_x =
  \mathcal{X}_x^{co(\mathcal{H})}$. Also,   $\overline{co(\mathcal{H})}$ is convex and closed.
  If there is
  $\{H_n\}\subset \overline{co(\mathcal{H})}$ such that
  $H_n \to H$ almost surely in $\mathbb{P}$, then
  $H_n\to H$ in probability and
  $H\in \overline{co(\mathcal{H})}$.    Therefore, we  apply the procedures in \cite[Proposition 3.2,
  Theorem 4.1]{CK01} to obtain the existence of
  $(\hat G, \hat H, \hat a, \hat X) \in \mathcal{G} \times
  \overline{co(\mathcal{H})} \times [0,\infty)
  \times \mathcal{X}_x$
  satisfying \eqref{eq:1}-\eqref{eq:GX}, the optimality of
  \eqref{eq:gha}, and the representation
  \begin{equation}
    \label{eq:Vr1}
    V(x) = \mathbb{E}[\hat G\hat X] =
    \inf_{a \ge 0} \{xa + \inf_{\mathcal{G} \times \overline{co(\mathcal{H})}}
    \mathbb{E} [(G- a H)^+]\}.
  \end{equation}Specifically, we replace the two probability density
  sets  in \cite{CK01} by the
  $L^{1}$-bounded sets $\mathcal{G}$ and $\mathcal{H}$ for our problem, and their $\mathcal{H}_x$ by $\overline{co(\mathcal{H})}$.
  At the infimum, $V(x)$ in \eqref{eq:Vr1} becomes (see \cite[Proposition 3.2(i)]{CK01})
  \begin{equation}
    \label{eq:Vr4}
    V(x) = x \hat a +
    \mathbb{E} [(\hat G - \hat a \hat H)^+]\}.
  \end{equation}
  Note that $\hat H$ belongs to $\overline{co(\mathcal{H})}$ but not
  necessarily to $co(\mathcal{H})$. Nevertheless, there
  exists a sequence $\{H_n\} \subset co(\mathcal{H})$
  satisfying     $H_n \to \hat H$  in probability.
  By the fact that any subsequence contains almost surely
  convergent subsequence, and together with
  the Dominated Convergence Theorem, it follows that
  $\mathbb{E}[(\hat G - \hat a H_n)^+] \to
  \mathbb{E}[(\hat G - \hat a \hat H)^+]$,
  and hence, representation \eqref{eq:Vr} follows.

  Next, for arbitrary $x_1, x_2 \ge 0$, the inequality

  \begin{align*}
    &\frac{1}{2}(V(x_1) + V(x_2)) \\
    &=   \frac 1 2 \Big(
    \inf_{\substack{a\ge 0\\ (G,H)\in \mathcal{G}\times co(\mathcal{H})}}
   \!\!   \!\! \mathbb{E} [x_1 a + (G  - aH)^+] +
    \inf_{\substack{a\ge 0\\ (G,H)\in \mathcal{G}\times co(\mathcal{H})}}
   \!\!   \!\!  \mathbb{E} [x_2 a + (G  - aH)^+] \Big)
    \\ &    \le
    \inf_{\substack{a\ge 0\\ (G,H)\in \mathcal{G}\times co(\mathcal{H})}}   \!\!   \!\!
    \mathbb{E} \Big[\frac 1 2 (x_1+x_2) a + (G  - aH)^+\Big] \\
& =  V \Big(\frac {x_1 + x_2} 2 \Big)
\end{align*}
  implies the concavity of $V(x)$.   The boundedness together
  with concavity yields continuity.

  Finally, we observe that
  if  $(\hat G, \hat H, \hat a, \hat X) \in\! \mathcal{G} \times
    \overline{co(\mathcal{H})} \times [0,\infty) \times
    \mathcal{X}_x$
    satisfies \eqref{eq:1}-\eqref{eq:GX}, then
    $(\hat G, \hat H, \hat a,  \widetilde X) \in\! \mathcal{G} \times
    \overline{co(\mathcal{H})} \times [0,\infty) \times
    \mathcal{X}_x$ with
    \[\widetilde X :=
    I_{\{\hat G > \hat a \hat H\}} +  \widetilde B I_{\{ \hat G=\hat a \hat H
      \}}, \hbox{ where } \widetilde B :=  \mathbb E [B | \sigma(\mathcal{G}
    \cup \mathcal{H})], \]
    also satisfies \eqref{eq:1}-\eqref{eq:GX}. Hence, $\hat X$ and $B$ can be chosen to be
$\sigma(\mathcal{G} \cup \mathcal{H})$-measurable.
\end{proof}

Comparing to the similar result by Cvitani{\'c} and Karatzas
\cite{CK01}, we have improved the representation of $V(x)$ in
\eqref{eq:Vr}, where the minimization in $H$ is conducted over the
smaller set $co(\mathcal{H})$, instead of $\mathcal{H}_x$. This will
be useful for our
application to the outperformance portfolio optimization (see Section \ref{sect-finance}) since it is easier to identify and work with the set $co(\mathcal{H})$ in a financial market. Moreover,   the minimizer $\hat a$ in Theorem~\ref{thm:rnp} above      belongs to $[0,\infty)$, rather than  $(0,\infty)$ according to
Proposition~3.1 and Lemma 4.3 in  \cite{CK01}.  In Appendix \ref{app-a0}, we provide  an example
where $\hat{a}=0$ as well as a sufficient condition for
$\hat{a}>0$.

We recall from
Lemma~\ref{prop:H}
that $V(x)$ of  \eqref{eq:rnp}  is invariant to replacing
$\mathcal{H}$ with any larger set $\mathcal{\tilde H}$
such that
$\mathcal{H} \subset \mathcal{\tilde H} \subset \mathcal{H}_x$. In
Theorem \ref{thm:rnp}, we observe that \eqref{eq:Vr}
also stays valid even if $co(\mathcal{H})$ is replaced
by any larger set
$\mathcal{\tilde H}$ such that
$co(\mathcal{H}) \subset \mathcal{\tilde H} \subset
\mathcal{H}_x$. However, the same does not hold if $co(\mathcal{H})$
is replaced by the original \emph{smaller} set
$\mathcal{H}$. We illustrate this technical
point in Example \ref{exm:coh} of Appendix~\ref{app-a4}.

It is also interesting to note that, one can take
$\mathcal{\tilde H}$ as the bipolar of $\mathcal{H}$
without changing the objective value,
which turns out
to be the smallest convex, closed, solid set containing $\mathcal{H}$
by the biploar theorem (see Theorem 1.3 of \cite{BS99}). To see this, if we denote the polar
of $\mathcal{A}\subset L^{0,+}$ by
$\mathcal{A}^o := \{X\in L^{0,+} : \mathbb{E}[AX] \le 1,
\forall A\in \mathcal{A}\}$, and
$x\mathcal{A} = \{xA: A\in \mathcal{A}\}$, then
$$\mathcal{X}_x^{\mathcal{H}} = (x \mathcal{H}^o) \cap
\mathcal{X} \subset x \mathcal{H}^o
\hbox{ and } \mathcal{H}_x = x (\mathcal{X}_x^{\mathcal{H}})^o
\supset x (x\mathcal{H}^o)^o = \mathcal{H}^{oo} \supset
\overline{co(\mathcal{H})}.$$
Precisely, the last inclusion
$ \mathcal{H}^{oo} \supset \overline{co(\mathcal{H})}$ above is due to the bipolar theorem. Moreover,
$\overline{co(\mathcal{H})}$ may be not solid, and strictly
smaller than the bipolar $\mathcal{H}^{oo}$, see Example~\ref{exm:pr}.

\subsection{On the Equivalence of Randomized and Pure Hypothesis Testing}\label{sect-equivhypo}
According to Theorem \ref{thm:rnp}, if the random variable $B$  in
\eqref{eq:1} can be assigned as an indicator
function satisfying \eqref{eq:1} - \eqref{eq:GX}, then  the
associated solver $\hat X$ of \eqref{eq:1}
will also be an indicator,  and
therefore, a \emph{pure test}! This leads to an
interesting question: when does a pure test solve the
randomized composite hypothesis testing problem?

Motivated by this, we define  the pure composite hypothesis testing problem:
\begin{align}&
  V_1(x) :=
  \sup_{X\in \mathcal{I}} \inf_{G\in \mathcal{G}} \mathbb{E}[ GX]\\
  \text{ subject to}\quad  & \sup_{H\in \mathcal{H}} \mathbb{E} [HX]\le x, \quad x>0.\end{align}
This is equivalent to solving
\begin{equation}
  \label{eq:np}
  V_1(x) = \sup_{X\in \mathcal{I}_x } \inf_{G\in \mathcal{G}} \mathbb{E} [GX],
\end{equation}where $\mathcal{I}_x := \{X\in \mathcal{I}: \mathbb{E}[H X] \le x, \
\forall H \in \mathcal{H}\}$ consists of all the candidate pure tests.

From their definitions, we see that  $V(x) \ge V_1(x)$. However, one cannot expect $V_1(x) = V(x)$ in general, as seen in the next simple example  from  \cite{LSYproceedings}.
\begin{example}\label{exm:pr}
  Fix $\Omega = \{0,1\}$ and $\mathcal{F} = 2^\Omega$, with
  $\mathbb{P}\{0\} = \mathbb{P}\{1\} = 1/2$. Define the collections
  $\mathcal{G} = \{G: G(0) = G(1) = 1\}$, and
  $\mathcal{H} = \{H: H(0) = 1/2,  H(1) = 3/2\}$. In this simple setup, direct computations yield that
  \begin{enumerate}
  \item For the randomized hypothesis testing,   $V(x)$  is given  by
    \begin{equation}
      \label{eq:eV}
      V(x) = \left\{
        \begin{array}{ll}
          \mathbb{E}[4xI_{\{0\}}] = 2x, & ~\text{if} ~~ 0\le x<1/4;\\
          \mathbb{E}[I_{\{0\}} + \frac{4x-1}{3} I_{\{1\}}] =
          \frac{2x+1}{3}, &  ~\text{if} ~~ 1/4\le x<1;\\
          \mathbb{E}[1] = 1, &  ~\text{if} ~~ x\ge 1.
        \end{array}\right.
    \end{equation}
  \item  For the pure hypothesis testing,
    $V_1(x)$ is given  by
    \begin{equation}
      \label{eq:eV1}
      V_1(x) = \left\{
        \begin{array}{ll}
          \mathbb{E}[0] = 0, & ~\text{if} ~~ 0\le x<1/4;\\
          \mathbb{E}[I_{\{0\}}] = \frac 1 2, & ~\text{if} ~~ 1/4\le x<1;\\
          \mathbb{E}[1] = 1, &  ~\text{if} ~~ x\ge 1.
        \end{array}\right.
    \end{equation}
  \end{enumerate}
  In the above, the inequality $V_1(x) < V(x)$ holds
  almost everywhere in $[0,1]$.  In fact, $V_1(x)$ is not
  concave   and continuous, while $V(x)$ is. $\Box$
\end{example}

\begin{remark}\label{remark:concavemajor}In Example \ref{exm:pr},  $V(x)$  turns out to be the smallest concave majorant of $V_1(x)$. However, this is not always true. We provide a counter-example in  Appendix \ref{sect-counter-major}.\end{remark}

If there is a pure test that solves both the pure and randomized
composite hypothesis testing problems, then the equality
$V_1(x)=V(x)$ must follow.  An important question is: when does this
phenomenon of equivalence occur?

\begin{corollary}
  \label{lem:B1}
  Let
  $(\hat G, \hat H, \hat a, \hat X) \in \mathcal{G} \times
  \overline{co(\mathcal{H})} \times [0,\infty) \times \mathcal{X}_x$
  be given
  by Theorem \ref{thm:rnp}. Then, $B$ in \eqref{eq:1} must satisfy
  \begin{enumerate} [(i)]
  \item If $\mathbb{E}[\hat H I_{\{\hat G> \hat a \hat H\}}] = x$,
    then $B=0$.
  \item If $\mathbb{E}[\hat H I_{\{\hat G \ge \hat a \hat H\}}] = x>\mathbb{E}[\hat H I_{\{\hat G> \hat a \hat H\}}]$,
    then $B=1$.
  \end{enumerate}
\end{corollary}
\begin{proof} In view of  the existence of $\hat X$  in  Theorem~\ref{thm:rnp} and its form in \eqref{eq:1}, $B$ as specified in each case above is the unique choice that  satisfies  $\mathbb{E}[\hat H \hat X] = x$ (see \eqref{eq:2}).
\end{proof}

Corollary \ref{lem:B1} presents two  examples where the optimal test $\hat{X}$ is indeed a pure test. In the remaining case where $\mathbb{E}[\hat H I_{\{\hat G \ge \hat a \hat H\}}] >  x
> \mathbb{E}[\hat H I_{\{\hat G > \hat a \hat H\}}]$, $B$ is a random variable taking value in $[0,1]$. When $\G$ and $\H$ are singletons, we have the following.

\begin{corollary}
  \label{lem:B3}
  Assume that $\mathcal{G} = \{\hat G\}$
  and $\mathcal{H} = \{\hat H\}$ are singletons, and
 \[\mathbb{E}[\hat H I_{\{\hat G \ge \hat a \hat H\}}] >  x
  > \mathbb{E}[\hat H I_{\{\hat G > \hat a \hat H\}}],\]
  Then, $B$ in \eqref{eq:1} can be taken as the constant
  \begin{equation}
    \label{eq:B}
    B_0 := \frac{x-\mathbb{E}[\hat H I_{\{\hat G>\hat a \hat H\}}]}
    {\mathbb{E}[\hat H I_{\{\hat G = \hat a \hat H\}}]}>0.
  \end{equation}
\end{corollary}
\begin{proof}
  This follows from direct computation to  verify \eqref{eq:1}-\eqref{eq:GX} in Theorem \ref{thm:rnp}.
\end{proof}

In Corollary~\ref{lem:B3}, we see that when $\mathbb{E}[\hat H I_{\{\hat G \ge \hat a \hat H\}}] >  x
> \mathbb{E}[\hat H I_{\{\hat G > \hat a \hat H\}}]$, the choice of $B = B_0 \in (0,1)$ yields a non-pure test $\hat X$ (see \eqref{eq:1}).
Nevertheless, our next lemma shows that, under an additional
condition,  one can alternatively  choose an indicator in place of $B$
and obtain a pure test.

\begin{lemma}
  \label{lem:npg}
  Assume that $\mathcal{G} = \{\hat G\}$
  and $\mathcal{H} = \{\hat H\}$ are singletons, and there exists an
  $\mathcal{F}$-measurable random variable
  $Y$, such that the function
  \begin{equation}
    \label{eq:g}
    g(y) = \mathbb{E}[ \hat H I_{\{Y < y\}}], \quad \forall y\in \mathbb{R},
  \end{equation}
  is continuous. Then there exists  a pure test $\hat X$ that solves
  both  problems \eqref{eq:rnp} and \eqref{eq:np}.
\end{lemma}
\begin{proof}
  If $(\hat G, \hat H, \hat a)$ satisfies either (i) or (ii) of Corollary~\ref{lem:B1},
  then Corollary~\ref{lem:B1} implies that $\hat X$ must be an indicator.
  Next, we discuss the other case: when
  $(\hat G, \hat H, \hat a)$ satisfies
  $\mathbb{E}[\hat H I_{\{\hat G \ge \hat a \hat H\}}] >  x
  > \mathbb{E}[\hat H I_{\{\hat G > \hat a \hat H\}}]$.
  Define a function
  $g_1(\cdot)$ by
  $$g_1(y) =  \mathbb{E}[ \hat H I_{\{\hat G = \hat a \hat H\} \cap \{Y
    < y\}}].$$
  Note that $g_1(\cdot)$ is right-continuous since, for any $y\in \mathbb{R}$,
  $$
  \begin{array}{ll}
    |g_1(y+\varepsilon) - g_1(y)|
    & = \displaystyle \mathbb{E}[ \hat H I_{\{\hat G = \hat a \hat H\} \cap \{
      y \le Y < y+ \varepsilon\}}] \\
    & \le \displaystyle  \mathbb{E}[ \hat H I_{\{y \le Y < y +
      \varepsilon\}}] \\
    & = g(y+\varepsilon) - g(y) \to 0, \ \hbox{ as } \varepsilon \to 0^+
  \end{array}
  $$
  by the continuity of $g(\cdot)$.  Similar arguments show that
  $g_1(\cdot)$ is also left-continuous.
  Also, observe that $$\lim_{y\to -\infty} g_1(y) = 0, \hbox{ and }
  \lim_{y\to \infty} g_1(y) =  \mathbb{E}[ \hat H I_{\{\hat G = \hat a
    \hat H\}}] > x - \mathbb{E}[ \hat H I_{\{\hat G > \hat a \hat
    H\}}].$$
  Therefore, there exists $\hat y\in \mathbb{R}$ satisfying
  \begin{equation}
    \label{eq:yhat}
    g_1(\hat y) =  x - \mathbb{E}[ \hat H I_{\{\hat G > \hat a \hat
      H\}}].
  \end{equation}
  Now, we can simply set
  \begin{align}\label{Xhat}
    \bar X =  I_{(\{\hat G = \hat a \hat H\} \cap \{Y < \hat y\}) \cup
      \{\hat G > \hat a \hat H\}} =
    I_{\{\hat G > \hat a \hat H\}} +    I_{\{Y < \hat y\}} \cdot
    I_{\{ \hat G=\hat a \hat H
      \}}.
  \end{align}
  One can directly verify that the above
  $\bar X$ belongs to $\mathcal{X}_x$ and satisfies
  \eqref{eq:1}, \eqref{eq:2}, and \eqref{eq:GX} with
  the choice of $B =  I_{\{Y < \hat y\}} $.
\end{proof}

In Lemma \ref{lem:npg}, if  the random variable $Y$ is \emph{continuous}, i.e. its cumulative distribution function (c.d.f.) $F_Y(y) = \mathbb{P}(Y< y)$  is continuous, then $g(y)$ in  \eqref{eq:g} must also be continuous and the result applies. Note that  $Y$ does not need  to be independent of $\G$ and $\H$.  Next, we establish a similar result for  the case where $\G$ and $\H$ are not singletons.

\begin{lemma}
  \label{lem:npg2}  Assume there exists
  a $\mathcal{F}$-measurable random variable $Y$
  independent of $\sigma(\mathcal{G} \cup \mathcal{H})$ with
  continuous cumulative distribution function. Then there exists  a
  pure test $\bar X$ that solves
  both  problems \eqref{eq:rnp} and \eqref{eq:np}.
\end{lemma}
\begin{proof} First, we define $U = F_{Y}(Y)$, which is uniformly distributed due to the continuity of  $F_{Y}(\cdot)$, and
  independent of $\sigma(\mathcal{G} \cup \mathcal{H})$.
  Let  $(\hat G, \hat H, \hat a, \hat X)
  \in \mathcal{G}
  \times {\overline{co(\mathcal{H})}} \times [0,\infty) \times \mathcal{X}_x$
  be chosen as of Theorem~\ref{thm:rnp}, where $\hat X$ is measurable
  with respect to $\sigma(\mathcal{G} \cup \mathcal{H})$.
  Then, we will show that the indicator
  \begin{equation}
    \label{eq:Xbar1}
    \bar X := I_{\{U<\hat X\}},
  \end{equation}
  also solves the problem \eqref{eq:np} by checking \eqref{eq:1}, \eqref{eq:2}, and \eqref{eq:GX}. To this end,  $\bar X$ satisfies \eqref{eq:1} since it admits the form
    $$\bar X = I_{\{\hat G > \hat a \hat H\}} +   I_{\{U<B\}} I_{\{ \hat G=\hat a \hat H
      \}},$$
       with the same $B$ in  \eqref{eq:1}. Next, for any random variable
       $M\in \mathcal{G} \cup \mathcal{H}$,   we use  the tower property to obtain
\begin{align*}
 \mathbb E[M \bar X] &= \mathbb E [M I_{\{U<\hat X\}}]
 \\ &=
 \mathbb E\big[ \mathbb E[M I_{\{U<\hat X\}}| \sigma(\mathcal{G} \cup \mathcal{H})] \big] \\
 &= \mathbb E\big[ M \mathbb E[ I_{\{U<\hat X\}}| \sigma(\mathcal{G} \cup \mathcal{H})] \big] \\
 &= \mathbb E [M \hat X].
\end{align*} In the last  equality, we have used the fact that  $\mathbb E[I_{\{U<c\}}] = \mathbb P\{U<c\} = c$ for $c\in [0,1]$ together with the measurability of $\hat X$ with respect to $\sigma(\mathcal{G} \cup \mathcal{H})$, which yields that $\hat X = \mathbb{E}[ I_{\{U<\hat X\}} |  \sigma(\mathcal{G} \cup \mathcal{H})]$ almost surely in $\mathbb{P}$.

Hence, we have
  $\mathbb{E}[H \hat X] = \mathbb{E}[ H \bar X]$ and
  $\mathbb{E}[G \hat X] = \mathbb{E}[ G \bar X]$ for all
  $(G, H) \in \mathcal{G}\times \mathcal{H}$, and this implies
  $\bar X$ satisfies both \eqref{eq:2} and \eqref{eq:GX}.
  As a consequence, the indicator
  $\bar X$ indeed solves both pure and randomized test by the definition.
\end{proof}

The fact that  an independent random variable appears in the equivalence between pure and randomized testing problems is quite natural. Indeed,
in hypothesis testing, statisticians may interpret the randomized test
by a pure test combined with  an independent
random variable drawn from a  uniform distribution. In
Lemma \ref{lem:npg2}, we have introduced the uniform random variable
$F_Y(Y)$  to the same effect.

Next, we summarize a number of  sufficient conditions that are amenable for
verification.

\begin{theorem}\label{thm:np} Suppose that one of the following conditions is satisfied:
  \begin{enumerate}
  \item [{\rm (C1)}] $\mathcal{G}$ and  $\mathcal{H}$
    are singletons, and there exists
    an $\mathcal{F}$-measurable random variable with  a continuous c.d.f.  with respect to $\mathbb{P}$,
  \item [{\rm (C2)}]There exists a continuous $\mathcal{F}$-measurable random variable independent of $\mathcal{G}$ and $\mathcal{H}$,
  \item [{\rm (C3)}]
    For all $0<x< \sup_{H\in \mathcal{H}} \mathbb{E}[H]$,
    its associated optimal triplet
    $(\hat G, \hat H, \hat a)
    \in \mathcal{G} \times {\overline{co(\mathcal{H})}} \times
    [0,\infty)$ given by Theorem~\ref{thm:rnp} satisfies
    $\mathbb{P}\{\hat G = \hat a \hat H\} = 0$.
  \end{enumerate}
  Then  $V_1(x) = V(x)$, and there exists an
  indicator function $\hat X$ that solves  problems
  \eqref{eq:rnp} and \eqref{eq:np} simultaneously. Furthermore,
  $x\mapsto V_1(x)$ is continuous, concave, and non-decreasing.
\end{theorem}
\begin{proof}
  In view of Lemma~\ref{lem:npg} and Corollary~\ref{lem:B1},
  we conclude $V_1(x) = V(x)$ under either of  (C1) or (C3).
  On the other hand, (C2) also implies $V_1(x) = V(x)$ due to
  Lemma~\ref{lem:npg2}.

  Since $V_1(x) = V(x)$, $V_1(x)$ inherits from $V(x)$ in Theorem
  \ref{thm:rnp} to be  continuous, concave, and non-decreasing.
\end{proof}

Note that condition  (C1)  in Theorem~\ref{thm:np} is slightly
stronger than \eqref{eq:g}. However, these are convenient to be used
to solve quantile hedging in the financial market. Comparing conditions
(C1) and (C2) in  Theorem~\ref{thm:np},  (C2) works for cases when $\mathcal{G}$ and $\mathcal{H}$ are not singletons, but it requires  that the continuous random variable be independent of $\mathcal{G}$ and $\mathcal{H}$. In contrast, (C1) does not require such an independence.

\begin{remark}\label{remark:C1}
  As it turns out, one cannot remove the  independence requirement on
  the continuous random variable in {\rm (C2)}  of Theorem~\ref{thm:np}.
  For the purpose of the illustration,
  we  provide a counter-example in Appendix~\ref{sec:dig}.
\end{remark}

\begin{remark}\label{remark:proceed}
In this section, our analysis   is conducted under the framework  $L^{0,+}(\Omega, \mathcal{F}, \mathbb{P})$ with
      topology given      by convergence in probability. This differs from  that in the authors' short proceedings paper \cite{LSYproceedings}, which summarized a small number of similar results under the framework $L^{1,+}(\Omega, \mathcal{F}, \mathbb{P})$ with $\mathbb{P}$-a.s. convergence. Moreover, the current paper has revised the main results, especially Theorem  \ref{thm:rnp} and Theorem \ref{thm:np}, and provided new lemmas as well as  rigorous proofs.
\end{remark}

\section{Outperformance Portfolio Optimization}\label{sect-finance}
We now discuss a portfolio optimization problem whose objective is to maximize the
probability of outperforming a random benchmark.  Applying our preceding analysis and the generalized  Neyman-Pearson lemma,
we will examine the problem  in both complete and incomplete markets.

\subsection{Characterization via Pure Hypothesis Testing}\label{sec:31}
We fix $T>0$ as the investment  horizon and let $(\Omega, \mathcal{F},
(\F_t)_{0\leq t\leq T},\mathbb{P})$ be a filtered complete
probability space  satisfying the usual conditions.
The market consists of a liquidly traded risky asset and a riskless
money market account.  For notational simplicity, we assume a  zero
risk-free interest rate, which amounts to working with cash flows
discounted by the risk-free rate.  We model  the risky asset   price
by a $\mathcal{F}_t$-adapted locally bounded
non-negative  semi-martingale process $(S_t)_{t\ge0}$.

The class of  Equivalent Local Martingale Measures (EMMs),
denoted by $\mathcal Q$, consists of all probability
measures $\mathbb{Q}\sim \mathbb P $ on
$\mathcal{F}_T$ such that the stock price $S$
is a $\mathbb{Q}$-local martingale.
We assume no-arbitrage in the sense of no free lunch
with vanishing risk (NFLVR). According
to \cite{DS94} (or Chapter 8 of \cite{DS06}),
this is a necessary and sufficient condition to have a non-empty
set $\mathcal Q$ for the locally bounded semi-martingale process.
We denote the  associated set of Radon-Nikodym
densities by
\[\mathcal{Z} := \Big\{\frac{d\mathbb{Q}}{d\mathbb{P}}: \mathbb{Q} \in \mathcal Q\Big\}.\]

Given an  initial capital $x$ and  a   self-financing trading strategy $(\pi_u)_{0\le u\le T}$ representing the number of shares in $S$, the investor's trading wealth process satisfies
\begin{equation}
  \label{eq:X}
  X^{x,\pi}_t = x + \int_0^t {\pi_u} d S_u.
\end{equation}
Each admissible trading strategy $\pi$ is a $\mathcal{F}_t$-progressively measurable process, such that the stochastic integral  $\int_0^t {\pi_u} d S_u$ is well-defined and  $X^{x,\pi}_t \ge 0$, $\forall t\in [0,T], \ \mathbb{P}-a.s.$ See  Definition 8.1.1
of \cite{DS06}. We denote the set of all admissible strategies by   $\mathcal{A}(x)$.

The benchmark is modeled by   a non-negative random terminal variable
$F \in \F_T$.   The smallest super-hedging price  (see e.g.
\cite{ElKaroui1995}) is defined as
\begin{equation}  \label{eq:F0}  F_0 := \sup_{Z\in \mathcal{Z}} \mathbb{E}[Z F],\end{equation}
which is assumed to be finite.  In other words, $F_0$ is the smallest capital needed for  $\mathbb{P}
\{X^{x, \pi}_T \ge F\} = 1$ for some strategy $\pi \in \mathcal{A}(x)$.
Note that with  less initial capital $x<F_0$ the success probability
$\mathbb{P}\{X^{x,\pi}_T\ge F\}<1$ for all $\pi\in \mathcal{A}(x)$.

Our objective is to maximize over all admissible trading strategies the success probability with  $x<F_0$. Specifically, we solve the optimization problem:
\begin{align}
  \label{eq:V}
  \widetilde V(x) &:= \sup_{x_1 \le x} \sup_{\pi\in \mathcal{A}(x_1)}
  \mathbb{P} \{X^{x_1,\pi}_T \ge {F}\}\\
  &= \sup_{\pi \in
    \mathcal{A}(x)} \mathbb{P}\{X^{x,\pi}_T \ge {F}\}, \qquad
  x\ge 0.  \label{eq:V2}
\end{align}
The second equality \eqref{eq:V2} follows from the monotonicity of the mapping $x \mapsto \sup_{\pi \in \mathcal{A}(x)} \mathbb{P}\{X^{x,\pi}_T\ge {F}\}$.  Clearly, $\widetilde V(x)$ is increasing in $x$. Moreover, if $F>0$ $\P$-a.s., then  $\widetilde V(0) = 0$ due to the non-negative wealth constraint.

\textbf{Scaling property.}  If the benchmark is scaled by a factor $\beta\ge 0$, then what is its effect to the success probability, given any fixed initial capital?
To address this, we first define \begin{align}\label{Vbb}\widetilde{V}(x;\beta) := \sup_{\pi \in
    \mathcal{A}(x)} \mathbb{P}\{X^{x,\pi}_T \ge \beta{F}\}. \end{align}

\begin{proposition} \label{prop-scale1}
  For any fixed $x> 0$, the
  success probability has the following properties:
  \begin{align}
    &(i)~~\text{The mapping }\,\beta \mapsto \widetilde{V}(x;\beta) \text{ is non-increasing for } \beta \ge 0, \notag \\
    &(ii)\,~\label{scale}\widetilde{V}(\beta x;\beta) = \widetilde{V}(x;1), \quad \text{ for }~ \beta \ge 0, \qquad\qquad\qquad\qquad\qquad\qquad\qquad\qquad\qquad\qquad\textcolor[rgb]{1.00,1.00,1.00}{.........}\\
    &(iii)~\lim_{\beta \to \infty} \widetilde{V}(x;\beta)= \P\{F=0\},\label{limitsbeta}\\
    &(iv)\,~\label{success1}\widetilde{V}(x;\beta)=1,~ \text{ for }\quad 0\le \beta \le \frac{x}{F_0}.\end{align}
\end{proposition}
\begin{proof}  First, we observe that $\widetilde{V}(x;\beta)= \sup_{\pi \in
    \mathcal{A}(x/\beta)} \mathbb{P}\{X^{x/\beta,\pi}_T \ge{F}\}$.
  Therefore, increasing $\beta$  means reducing the initial capital for beating the same benchmark $F$, so (i) holds.   Substituting $x$ with $\beta x$, we obtain  (ii). To show (iii), we write  \begin{align}\widetilde{V}(x;\beta) &=  \sup_{\pi \in
      \mathcal{A}(x)} \big( \mathbb{P}\{X^{x,\pi}_T \ge \beta{F}, F=0\} + \mathbb{P}\{X^{x,\pi}_T \ge \beta{F}, F>0\}\big) \notag\\
    &= \mathbb{P}\{  F=0\}  + \sup_{\pi \in
      \mathcal{A}(x)} \mathbb{P}\{X^{x,\pi}_T \ge \beta{F}, F>0\}. \label{Fge0}\end{align}
  Focusing on the  second term of  \eqref{Fge0}, it suffices to consider an
  arbitrary strictly positive benchmark ${F_+}>0$. We deduce from  (i)
  and $\widetilde V(0)=0$ that
  \[\lim_{\beta \to \infty}\sup_{\pi \in
    \mathcal{A}(x/\beta)} \mathbb{P}\{X^{x/\beta ,\pi}_T \ge  F_+\} = \lim_{x \to 0}\sup_{\pi \in
    \mathcal{A}(x)} \mathbb{P}\{X^{x ,\pi}_T \ge  F_+\} = 0.\] This together with \eqref{Fge0} implies the
  limit \eqref{limitsbeta}.

  Lastly, when the initial capital exceeds the super-hedging price of $\beta$ units of
  $F$, i.e. $x\ge \beta F_0$,  the success probability
  $\widetilde{V}(x;\beta)=1$ and hence (iv) holds.
\end{proof}

In other words, for any initial capital $x$, the success probability
$\widetilde{V}(\beta x;\beta)$ stays constant whenever  the initial capital
and benchmark are simultaneously scaled by $\beta >0$. To see this,
suppose the optimal strategy for beat one unit of the benchmark $F$  is
$\pi^*_1$.  If the investor wants to outperform the benchmark $\beta F$, then he can trade using  the same strategy $\pi^*_1$ in $\beta$ separate accounts
and will achieve  the same level of success probability as in the single benchmark case.
Proposition \ref{prop-scale1} points out that this strategy is optimal for any $\beta>0$, and hence, there is no \emph{economy of scale}.

\begin{remark} For any fixed $x> 0$, the success  probability $\widetilde{V}( x;\beta)$  is not convex  or concave in $\beta$.  This can be easily inferred from the properties of $\widetilde{V}$ shown in Proposition \ref{prop-scale1}, and is illustrated in Figure \ref{fig1} below.
\end{remark}

Next, we show that the portfolio optimization problem \eqref{eq:V}  admits a dual representation  as a  pure hypothesis testing problem. Such a  connection was first pointed out  by F\"ollmer and Leukert \cite{FL99} in the context of quantile hedging.

\begin{proposition}\label{lem:npi} The value function $\widetilde
  V(x)$ of \eqref{eq:V} is equal to the solution of a pure hypothesis
  testing problem, that is, $\widetilde V(x) = V_1(x)$ where
  \begin{align}  \label{eq:lnpi2}
    V_1(x) = &\sup_{A\in \mathcal{F}_T} \mathbb{P}\{A\}      \\
    \hbox{subject to }
    &\sup_{Z\in \mathcal{Z}} \mathbb{E} [Z F I_A] \le x.
  \end{align}
  Furthermore, if there exists $\hat A\in \mathcal{F}_T$ that solves
  \eqref{eq:lnpi2}, then $\widetilde V(x) = \mathbb{P}\{\hat A\}$,  and
  the associated optimal strategy $\pi^*$  is a super-hedging strategy
  with $X^{x,\pi^*}_T \ge F I_{\hat A}$ \,$\P$-a.s.
\end{proposition}
\begin{proof} First,   if we set $\H = \{Z F:Z\in \mathcal{Z}   \}$ and
  $\G=\{1\}$, then the right-hand side of \eqref{eq:lnpi2} resembles
  the pure hypothesis testing problem in  \eqref{eq:np}.
  \begin{enumerate}
  \item First, we prove  that $V_1(x) \ge  \widetilde V(x)$.
    For an  arbitrary $\pi\in \mathcal{A}(x)$, define the success
    event  $A^{x,\pi} := \{X^{x,\pi}_T  \ge F\}$. Then,
    $\sup_{Z\in \mathcal{Z}} \mathbb{E} [Z F I_{A^{x,\pi}}] $ is the
    smallest amount needed to super-hedge $FI_{A^{x,\pi}}$. By the
    definition of $A^{x,\pi}$, we have that $X^{x,\pi}_T \ge F I_{A^{x,\pi}}$,
    i.e. the initial capital $x$ is sufficient to super-hedge   $F
    I_{A^{x,\pi}}$. This
    implies that $A^{x,\pi}$  is a candidate solution to $V_1$ since
    the constraint $x\ge \sup_{Z\in \mathcal{Z}} \mathbb{E} [Z F
    I_{A^{x,\pi}}]$ is satisfied. Consequently, for any
    $\pi \in \mathcal{A}(x)$, we have
    $V_1(x) \ge \mathbb{P}\{A^{x,\pi}\}$.
    Since
    $\widetilde V(x) = \sup_{\pi\in \mathcal{A}(x)}
    \mathbb{P}\{A^{x,\pi}\}$
    by \eqref{eq:V}, we conclude.

  \item Now, we show  the reverse inequality
    $V_1(x) \le \widetilde V(x)$. Let $A\in \mathcal{F}_T$ be an
    arbitrary set satisfying the constraint
    $\sup_{Z\in \mathcal{Z}} \mathbb{E} [Z F I_A] \le x$. This implies a  super-replication by some $\pi\in \mathcal{A}(x)$ such
    that $\mathbb{P}\{X^{x,\pi}_T \ge FI_A\}=1$. In turn, this yields $\mathbb{P}\{X^{x,\pi}_T \ge F\} \ge
    \mathbb{P}\{A\}$. Therefore, $\widetilde V(x) \ge \mathbb{P}\{A\}$ by
    \eqref{eq:V}. Thanks to the arbitrariness of $A$, $\widetilde V(x)
    \ge V_1(x)$
    holds.
  \end{enumerate}
  In conclusion, $\widetilde V(x) = V_1(x)$. Moreover,  if a set $\hat{A}$ satisfies that
  $\widetilde V(x) = \mathbb{P}\{\hat{A}\}$, then the corresponding strategy $\pi$ that
  super-hedges $FI_A$ is the solution of \eqref{eq:V}.
\end{proof}

Applying our analysis in Section \ref{sect-equivhypo}, we seek to connect  the  outperformance portfolio optimization problem, via its pure hypothesis testing representation,  to   a randomized hypothesis testing problem.   We first state  an explicit  example (see \cite{LSYproceedings}) where the outperformance portfolio optimization  is equivalent to the  pure  hypothesis testing  by Proposition \ref{lem:npi}, but not to  the randomized counterpart.

\begin{example}
  \label{exm:bin}
  Consider $\Omega = \{0,1\}$, $\mathcal{F} = 2^{\{0,1\}}$, and the real
  probability given by  $\mathbb{P}\{0\} = \mathbb{P}\{1\} =
  1/2$. Suppose stock price $S_t(\omega)$ follows one-period binomial
  tree:
  $$S_0(0) = S_0(1) = 2; \quad S_T(0) = 5, \ S_T(1) = 1.$$
  The benchmark $F=1$ at $T$. We will determine by direct computation the
  maximum  success probability given initial capital $x\ge 0$.  To this end, we notice that the  possible strategy with initial capital $x$ is  $c$ shares of stock plus $x-2c$ dollars of cash at $t=0$. Then, the terminal wealth $X_T$ is
  $$X_T = \left\{
    \begin{array}{ll}
      5c + (x - 2c) = x + 3c \quad & \omega = 0,\\
      c + (x-2c) = x-c \quad & \omega = 1.
    \end{array}\right.
  $$
  Due to the non-negative wealth constraint $X_T \ge 0$ a.s., we require that $-\frac x 3 \le c \le x$.
  Now, we can write $\widetilde V(x)$ as
  \begin{equation}\label{tilV}\widetilde V(x) = \max_{-\frac x 3 \le c \le x} \mathbb{P}\{X_T \ge 1\}
    = \frac 1 2  \max_{-\frac x 3 \le c \le x}  \Big( I_{\{x+3c \ge 1\}} +
    I_{\{x-c\ge 1\}} \Big).\end{equation}
  As a result, for different values of initial capital $x$ we have:
  \begin{enumerate}
  \item If $x<1/4$, then
    $$x+3c \le x + 3x = 4x <1$$
    and $$x - c \le x + \frac x 3 = \frac{4x}{3}<1/3,$$
    which implies both indicators are zero, i.e. $\widetilde V(x) = 0$.
  \item If $1/4\le x <1$, then we can take $c = 1/4$, which leads to
    $x+3c \ge 1$, i.e. $\widetilde V(x) \ge 1/2$.
    On the other hand, $\widetilde V(x) <1$. From this and  \eqref{tilV}, we conclude  that $\widetilde V(x) = 1/2$.
  \item If $x\ge 1$, then we can take $c=0$, and $\widetilde V(x) = 1$.
  \end{enumerate} With reference to the value functions $V(x)$ (randomized hypothesis testing) and $V_1(x)$ (pure hypothesis testing) from Example~\ref{exm:pr}, we conclude that $\widetilde V(x) = V_1(x) \neq V(x)$.
\end{example}

As in Theorem \ref{thm:np}, we now provide the sufficient conditions for the equivalence between the outperformance portfolio optimization
and the randomized hypothesis testing.

\begin{theorem}
  \label{thm:qhic}
  Suppose that one of two conditions below is satisfied:
  \begin{enumerate}
  \item $\mathcal{Z}$ is a singleton, and
    there exists a
    $\F_T$-measurable random   variable
    with continuous cumulative distribution function under
    $\mathbb{P}$;
  \item For all $a\in (0, \infty)$,
    the minimizer
    $\hat Z_a := \arg\min \mathbb{E}[xa + (1 - a ZF)^+]$
    satisfies
    $\mathbb{P}\{a \hat Z_a F = 1\} = 0$.
  \end{enumerate}Then,
  \begin{enumerate}[(i)]
  \item The value function $\widetilde{V}(x)$ of \eqref{eq:V} admits the representation:
    \begin{equation}
      \label{eq:Vic}
      \widetilde{V}(x) = \inf_{a\ge 0, Z\in \mathcal{Z}} \mathbb{E} [xa + ( 1 - aZF)^+].
    \end{equation}
  \item $\widetilde{V}(x)$  is continuous, concave, and
    non-decreasing in $x \in [0,\infty)$, taking values from the
    minimum $\widetilde V(0) = \mathbb{P}\{F = 0\}$ to the maximum $\widetilde V(x) =1$ for
    $x \ge F_0$.
  \end{enumerate}
\end{theorem}
\begin{proof}
  Proposition~\ref{lem:npi} implies that $\widetilde V(x)$ is
  equal to the value  $V_1(x)$ of pure testing problem
  with $\mathcal{H} :=  \{FZ: Z  \in \mathcal{Z}\}$  and
  $\mathcal{G} := \{1\}$.  Since conditions (1) and (2)
  satisfy (C1) and (C3) of Theorem~\ref{thm:rnp} respectively,
  this also implies that $V_1(x)$ of pure testing is equal to
  $V(x)$ of randomized testing.
  Note that $F_0 <\infty$ implies $\mathcal{H}$
  is $L^1$ bounded. Hence,  Assumption \ref{a:1} is
  satisfied along  with the convexity of
  the set $\mathcal{H}$. Thus, the
  representation  \eqref{eq:Vic} follows directly
  from  \eqref{eq:Vr}  of Theorem \ref{thm:rnp}.

  It remains to observe from  \eqref{eq:Vic}
  that $\widetilde{V}(x)\le 1$ by taking
  $a=0$. When $x=0$, the success event coincides with
  $\{F=0\}$, so the lower bound
  is $\widetilde{V}(0) = \P\{F=0\}$.
\end{proof}

\begin{remark} \label{rmk:1}
  Condition 1 of
  Theorem \ref{thm:qhic} together with \eqref{eq:1}  recovers
  Proposition 2.1 by  Spivak and Cvitani{\'c} \cite{SC99} with zero maintenance margin,
  (i.e.  $A = 0$ in Equation (2.30)  of \cite{SC99}).
  Furthermore, our pure test in  \eqref{Xhat} also reveals the structure of their set $E$.
\end{remark}

In Theorem \ref{thm:qhic}, condition 2  is typical  in the
quantile hedging literature (see e.g. \cite{FL99,melnikov}), but it can
be violated  even in the simple Black-Scholes model; see Section \ref{exm:gbm} (case 1).
In such cases, one may alternatively check condition 1 in order to
apply Theorem \ref{thm:qhic}.

In the following sections, we will discuss the applications of this result in both complete and incomplete diffusion market models.

\subsection{A Complete Market Model} \label{sec:cm}
Let $W$ be a standard Brownian motion on
$(\Omega, \mathcal{F}, \mathbb{P}, (\mathcal{F}_t)_{0\le t \le
  T})$. The financial market consists of a liquid  risky  stock and a riskless money market account.
For notational simplicity, we assume a zero interest rate, which amounts to expressing cash flows in the money market account numeraire.
Under the historical measure, the stock price evolves according to:
\begin{equation}
  \label{eq:sv1}
  \begin{array}{ll}
    d S_t &= S_t \sigma(S_t) \left(\,\theta(S_t) dt +  d W_t\,\right),
  \end{array}
\end{equation}
where  $\theta( \cdot )$ is the  Sharpe ratio function      and $\sigma(\cdot)$ is  the volatility function (see Karatzas and Shreve \cite[\S 1]{Karatzas1998a} for  standard conditions). For any   admissible strategy $\pi \in \mathcal{A}(x)$, the investor's wealth process associated with strategy $\pi$ and initial capital $x$ is given by
\begin{equation}
  \label{eq:Xvol1}
  d X^{x,\pi}_t = \pi_t S_t\sigma(  S_t) \left(\,\theta(S_t) dt +  d W_t\,\right).
\end{equation}

The investor's objective is to maximize the  probability of beating the benchmark $F = f(S_T)$
for some measurable function $f$.  Since a perfect replication is possible by trading $S$ and the money market account,  the market is complete, and there exists a unique EMM $\mathbb{Q}$  defined by
\[Z_t  := \frac{d\mathbb{Q}}{d\mathbb{P}}\big|_{\F_t} =  \exp \Big\{ -\frac 1 2 \int_0^t \theta^2(S_u)
du - \int_0^t \theta( S_u) d W_u \Big\}.
\]
Moreover, the super-hedging price is simply the risk-neutral value
$F_0  = \E^\mathbb{Q}[f(S_T)]$, which is  a
special case of \eqref{eq:F0}.
% \eqref{eq:V2}.
Given an initial capital $x<F_0$,  the
investor faces the optimization problem:
\begin{equation}
  \label{eq:Vvol1}
  \widetilde V(x)  = \sup_{\pi \in \mathcal{A}(x)}
  \mathbb{P} \{X_T^{x,\pi} \ge f(S_T)\}.
\end{equation}

\begin{proposition}
  \label{prop:gbm}
  $\widetilde V(x)$ is a continuous, non-decreasing, and concave function in $x$. It admits the dual representation:
  \begin{align}\label{rep_Vtil}\widetilde V(x) = \inf_{a\ge 0} \{ xa + \mathbb{E}[ (1 - a
    Z_T f(S_T))^+]\}.\end{align}
\end{proposition}
\begin{proof}
  First, Proposition \ref{lem:npi} implies $\widetilde V(x)=V_1(x)$ (the pure hypothesis testing).  Also, since $\mathcal{Z} =
  \{Z\}$ is a singleton, and $W_T$ has continuous c.d.f. with respect to $\mathbb{P}$,  the first condition of Theorem~\ref{thm:qhic} yields the equivalence of pure and randomized hypothesis testings, i.e. $\widetilde V(x) = V_1(x) = V(x)$.
\end{proof}

For computing the value of $\widetilde{V}(x)$  in this complete market model,
Proposition  \ref{prop:gbm} turns the  original stochastic control problem \eqref{eq:Vvol1} into a static  optimization (over $a \ge 0$)
in \eqref{rep_Vtil}.  In the dual representation, the expectation  can be interpreted as pricing a claim  under measure $\Q$, namely,
\[ q(a):= \mathbb{E}^\Q[  (  Z^{-1}_T - a f(S_T))^+]. \]
Hence,    $\widetilde{V}(x)$   is the  Legendre transform of the price function $q(a)$ evaluated at $x$.

\subsubsection{Benchmark Based on the Traded Asset}
\label{exm:gbm}
In this section, we assume that  $\theta$ and $\sigma$  are
constant, so $S$ is a geometric Brownian motion (GBM).  We consider a class of benchmarks of the form
$f(S_T) =  \beta S^p_T$, for $\beta >0, p\in \R$.   This includes the constant
benchmark ($p=0$), as well as those based on multiples of the traded
asset $S$ ($p=1$) and its power.

One interpretation of the power-type benchmarks is in terms of leveraged exchange traded funds (ETFs).   ETFs are investment funds liquidly traded on stock exchanges. They provide leverage, access, and  liquidity to investors for various asset classes, and typically involve strategies with a constant leverage (e.g. double-long/short).  They also serve as the benchmarks for fund managers.  Since its introduction in the mid 1990's, the ETF market has grown to over   1000 funds  with aggregate  value exceeding \$1 trillion.

Specifically, a long-leveraged ETF $(L_t)_{t\geq0}$ based on the underlying asset $S$ with a constant leverage factor $p\geq0$ is constructed by investing $p$ times the fund value  $p L_t$ in $S$ and  borrowing $(p-1)L_t$ from the bank.  The resulting fund price $L$ satisfies the SDE  (see \cite{AZ_ETF_SIAM10, JarrowETF10}):
\begin{align}
  dL_t  &=p L_t \Big(\frac{dS_{t}}{S_{t}}\Big)   =L_t\left(    p \theta \sigma\, dt+p\sigma dW_{t}\right).
  \label{ETFPrcProc}
\end{align}
As for a short-leveraged fund $p\leq0$, the manager shorts the amount $-pL_t$ of $S$, and keeps $(-p+1)L_t$ in the bank. The fund price $L$ again satisfies SDE \eqref{ETFPrcProc} with $p\le 0$. Hence,  $L$ is again a GBM and  can be expressed in terms of $S$ as\begin{align}
  \frac{L_t}{L_{0}}  &  = \left(\frac{S_t}{S_0}\right)^p  \exp\big\{ \frac{p(1-p)\sigma^{2}}{2}   t\big\} \label{LeveragEqn}.\end{align}
As a result, the objective to outperform a  $p$-leveraged ETF $L_T$  leads to a special example of the power benchmark $\hat{\beta} S^p_T$, with    $\hat{\beta} = {L_0}{S_0^{-p}}\exp\big\{\frac{p(1-p)\sigma^{2}}{2}  T\big\}$. In practice,  typical leverage factors are $p=1,2,3$ (long) and $-1, -2, -3$ (short).

More generally,  given any $(\beta, p)$, the risk-neutral price of the benchmark  $f(S_T) = {\beta} S^p_T$ is
\begin{align}\label{FGBM}F_0 = \beta S_0^p \exp\{\frac{\sigma^2}{2}p(p-1)T\}.\end{align}
Clearly, if $x\ge F_0$, the success probability is 1, so the challenge is to achieve the
outperformance using less initial capital.
Then,  a direct   computation using \eqref{rep_Vtil} and \eqref{FGBM} yields that
\begin{equation}
  \label{eq:rho1}
  \begin{array}{ll}
    \widetilde V(x) &= \inf_{a\ge 0}\limits \left\{xa + \mathbb{E}[( 1 - a F_0 \exp\left\{- \frac 1 2 (p\sigma - \theta)^2 T + (p\sigma - \theta)
        W_T\right\})^+] \right\}.
  \end{array}
\end{equation}
To solve for $\widetilde V(x)$, we  divide the problem into two cases:
\begin{enumerate}
\item If $p\sigma = \theta$, then  $ZF=F_0$ a.s., so condition 2  in Theorem \ref{thm:qhic} is violated, but condition 1 holds and is used.
  Consequently,  \eqref{eq:rho1} simplifies to
  \begin{equation}\label{vtil1}
    \begin{array}{ll}
      \widetilde V(x) & = \inf_{a\ge 0} \{xa + (1 - aF_0)^+\} =
      \left\{
        \begin{array}{ll}
          1, & \hbox{ if }  x\ge F_0 \\
          x/F_0, & \hbox{ if } x< F_0
        \end{array}\right.
    \end{array}
  \end{equation}
  and the corresponding minimizers are $\hat{a} = 0$  and  $\hat{a} =
  F^{-1}_0$ respectively.

\item If $p\sigma \neq \theta$, then  $\widetilde{V}(x) =1$ if $x\ge F_0$; otherwise, direct computations yield that
  \begin{align}\label{vtil2}
    \widetilde V(x)     & =  \inf_{a\ge 0}\limits  \left\{xa + \Phi(d_2(a;p\sigma - \theta )) - aF_0 \Phi(d_1(a; p\sigma - \theta)) \right\},\\
    & =  x \hat a + \Phi\big(d_2(\hat a; p\sigma - \theta)\big) - \hat aF_0 \Phi\big(d_1(\hat a; p\sigma - \theta)\big), \label{vtil3}\end{align}
  where $d_i$ are
  \begin{align}\label{dd}d_1(a; z) = \frac{-\ln (a F_0) -0.5  T  z^2}{|z| \sqrt T},
    \quad d_2(a;z) = \frac{-\ln (aF_0) + 0.5 T z^2}{|z| \sqrt T}.\end{align}
  Note that the infimum is reached at $\hat a$ which solves
  \begin{align}\label{hatHx}\mathbb{E}\left[F_0 \hat H I_{\{\hat a F_0 \hat H<1\}}\right] = x,\end{align}
  where
  $\hat H =  \exp\{-\frac 1 2 (p\sigma - \theta)^2 T + (p\sigma - \theta)
  W_T\}$. Let  $d \tilde{\mathbb{Q}} = \hat H d \mathbb{P}$, then \eqref{hatHx} implies that
  $$\tilde{\mathbb{Q}}\left\{ \hat H < \frac{1}{\hat a F_0}\right\} = \frac x F_0,$$
  which is equivalent to
  $$\tilde{\mathbb{Q}}\left\{ (p\sigma - \theta) (W_T + (\theta- p\sigma) T)  <
    -\ln (\hat a F_0) -\frac 1 2 (p\sigma -\theta)^2 T\right\} = \frac x F_0.$$
  Since $W_T + (\theta - p\sigma) T \sim \mathcal{N}(0, T)$
  under $\tilde{\mathbb{Q}}$, $\hat{a}$ is given by
  \begin{equation}
    \label{eq:rhoa2}
    \hat a = h\left( \Phi^{-1}(x/F_0)\right)
  \end{equation}
  where
  \[ h(y ) = \exp\left\{  - y |p\sigma - \theta |\sqrt{T} - 0.5 (p\sigma - \theta)^2 T - \ln F_0\right\}.\]
\end{enumerate}

In the above example, one can also compute the initial capital needed to
achieve a pre-specified success probability simply by inverting
$\tilde{V}(x)$ in \eqref{vtil2} and \eqref{vtil1}; see Fig. \ref{fig1}(a).  Also, note that $\tilde{V}(x)$ depends on $\beta$ via $F_0$ in \eqref{FGBM}. In Fig. \ref{fig1}(b) we see that  $\widetilde{V}(x;\beta )$ decreases from  1 and 0 as $\beta$ increases to infinity, which is consistent with the limit \eqref{limitsbeta}.

While the super-hedging price  $F_0$ is computed from
$\mathbb{Q}$, the maximal success probability $\widetilde{V}(x)$ is
based on the historical measure $\P$. In other words, as we vary the
Sharpe ratio $\theta$, the required initial capital $x$ to achieve a given
success probability will change,  but $F_0$ -- the cost to guarantee
outperformance --   remains unaffected (see Fig. \ref{fig1}(a)).

In Fig. \ref{fig2}, we look at the probability to outperform an ETF under different leverages. From \eqref{LeveragEqn}, we  note that $F_0 = \E^{\Q}[L_T]= L_0$.  Then, we apply formula \eqref{vtil3} to obtain the success probability $\widetilde V(x)$ for different values of capital $x$ and leverage $p$.  As shown, for every fixed $x$, moving the leverage $p$ further away from zero increases the success probability. In other words, for any fixed success  probability, highly  (long/short)  leveraged ETFs require lower initial capital for the outperformance portfolio. The comparison between long and short ETFs with the same magnitude of leverage $|p|$ depends on the sign of $\theta$. In particular, we observe from \eqref{vtil3} and \eqref{eq:rhoa2} that when $\theta=0$ the success probability $\widetilde{V}(x)$ is the same for $\pm p$, and  the surface $\widetilde{V}(x)$ is symmetric around $p=0$.

\begin{remark}In a related study,  F\"{o}llmer and Leukert \cite[Sect. 3]{FL99}
  considered quantile hedging   a call option in the Black-Scholes market.
  Their solution method involves first conjecturing the form of the success
  events   under two   scenarios.  Alternatively, one can also study the
  quantile hedging problem via randomized hypothesis testing. From
  \eqref{rep_Vtil} we can compute the maximal success probability from $\widetilde V(x) = \inf_{a\ge 0} \{ xa +
  \mathbb{E}[ (1 - a
  Z_T (S_T - K)^+)^+]\}$, which will yield exactly the same closed-form result  in \cite[Eq.(3.15),(3.27)]{FL99}. This approach  alleviates the need to
  {a priori}   conjecture the success events.
\end{remark}

\begin{figure}[th]\centering
\begin{subfigure}[{}]{\includegraphics[scale=0.6]{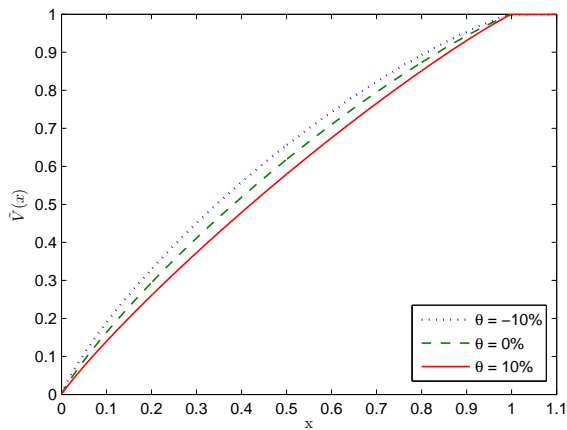}}\end{subfigure}
\begin{subfigure}[{}]{\includegraphics[scale=0.6]{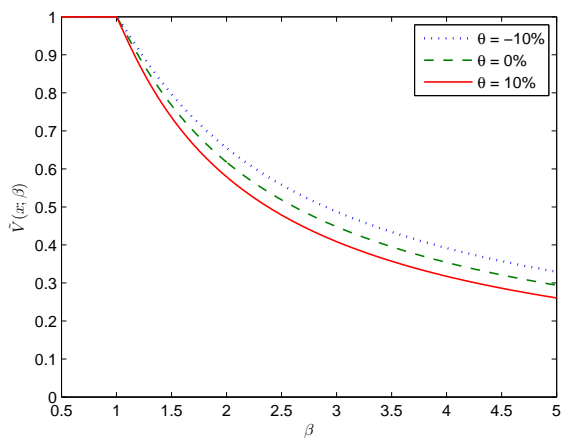}}\end{subfigure}
   \caption{\small{The benchmark is $F(S_T) = \beta S_T$, and the default parameters are $S_0=1$, $\sigma=30\%$, and $T=1$. (Top) With $\beta=1$, the maximum success probability $\widetilde{V}(x)$ increases with initial capital $x$, and plateaus at 1 when $x>S_0$. For any fixed success probability, a lower Sharpe ratio $\theta$ requires a lower  initial capital $x$.   (Bottom) With initial capital $x=1$, $\widetilde{V}(x;\beta )$ takes value 1 and then  decreases to 0 as $\beta$ increases to infinity.  Observe that $\widetilde{V}(x;\beta )$ is not simply convex or concave even over the range $[0.5,5]$ of $\beta$, and   converges to 0 as $\beta \to \infty$ according to \eqref{limitsbeta}.}}
  \label{fig1}
\end{figure}
 \clearpage

\begin{figure}[ht!]
  \centering
  \includegraphics[scale=0.57]{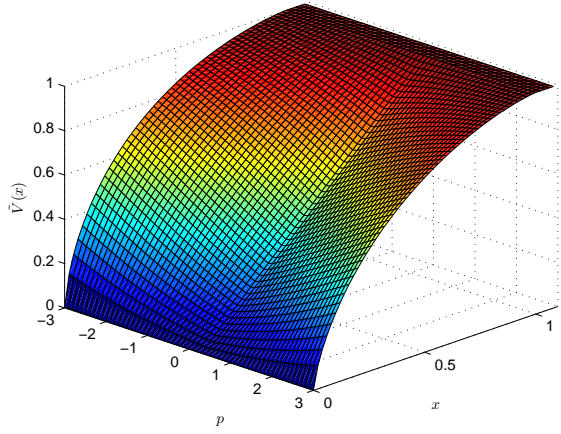}
  \caption{\small{Outperformance probability surface over leverage $p$ and initial capital $x$. For any fixed $x$, the probability $\widetilde{V}(x)$ increases as leverage $p$ increases/decreases from zero. This means that highly leveraged ETFs are easier benchmarks to beat.}}
  \label{fig2}
\end{figure}

\subsection{A Stochastic Factor Model}\label{sect:stochvol}
Let $(W, \hat{W})$ be a two-dimensional standard Brownian motion on
$(\Omega, \mathcal{F}, \mathbb{P}, (\mathcal{F}_t)_{0\le t \le
  T})$.  We consider a  liquid stock   whose price follows the SDE:
\begin{equation}
  \label{eq:sv}
  d S_t = S_t \sigma(Y_t) ( \theta(Y_t) dt +  d W_t),
\end{equation}
where  $\theta$  is the Sharpe ratio function, and the stochastic factor $Y$ follows \begin{equation}
  \label{eq:Y}
  d Y_t = b(Y_t) dt + c(Y_t) (\rho d W_t + \sqrt{1 -
    \rho^2} d \hat{W}_t).
\end{equation}
This is a standard stochastic factor/volatility model that can be found in, among others,  \cite{RomanoTouzi97,sircarzarivol}. The parameter $\rho \in (-1,1)$ accounts for the correlation between $S$ and $Y$.

With initial capital $x$ and   strategy $\pi\in \mathcal A(x)$, the wealth process satisfies
\begin{equation}
  \label{eq:Xvol}
  d X^{x,\pi}_t = \pi_t S_t\sigma(Y_t)  (  \theta(Y_t) dt + dW_t).
\end{equation}
Let $\Lambda$ be the collection of all $\mathcal{F}_t$
progressively measurable process $\lambda: (0, T)\times \Omega \to \mathbb{R}$
satisfying $\int_0^T\lambda^2_t dt <\infty$ $\mathbb{P}$-a.s., and denote the set of all Radon-Nikodym densities of equivalent martingale
measures by $\mathcal{Z} =
\{\tilde{Z}^{\lambda}_T: \lambda \in \Lambda \}$ where
\begin{align}\label{Ztil}\tilde{Z}^{\lambda}_T= \exp \Big\{ -\frac 1 2 \int_0^T\! \theta^2(Y_t)
  dt - \int_0^T\! \theta(Y_t) d W_t    -
  \frac 1 2 \int_{0}^T \!\lambda^2_t dt - \int_0^T\! \lambda_t d
  \hat{W}_t \Big\}.\end{align}
The process $\lambda$ is commonly referred to as the risk premium for the non-traded Brownian motion $\hat{W}$. In particular, the choice of  $\lambda =0$ results in the minimal martingale measure  (MMM) $\Q^0$ (see \cite{FollmerSchweizer1990}).

\subsubsection{The Role of the Minimal Martingale
  Measure}\label{exm:stochvol1}
Let us consider a benchmark of the form  $F = \beta S^\delta_T$, where
$\delta \in \{0,1\}$. This includes the constant and stock benchmarks.
Following \eqref{eq:V}, we consider the  optimization problem:
\begin{equation}  \label{eq:Vvol}  \widetilde V(x)  = \sup_{\pi \in \mathcal{A}(x)}
  \mathbb{P} \{X_T^{x,\pi} \ge \beta S^\delta_T \}.\end{equation}

\begin{proposition}\label{prop-stochvol}
  Suppose
  $c_1<|\theta(y) - \delta \sigma(y)|< c_2$ holds for all
  $(y, \delta) \in \mathbb{R}\times \{0,1\}$
  for     some positive constants $c_1$ and $c_2$. Then,
  the value function  $\widetilde V(x)$ in \eqref{eq:Vvol} is
  non-decreasing, continuous and concave function satisfying
  \begin{align}\label{vstochvol}\widetilde V(x) = \inf_{a\ge 0} \{ xa + \mathbb{E}[ (1 - a  \beta S_0^{\delta}\tilde{Z}^{{0}}_T)^+]\}.\end{align}
\end{proposition}
To show this, we will use the following result, which is a variation of \cite[Exercise
2.3]{JYC09} and the proofs  of (5.3) and (5.6) in \cite[p.19]{CK01}.

\begin{lemma}   \label{prop:compz}

  Let $B$ be a standard Brownian motion on
  $(\Omega, \mathcal{F}, \mathbb{P}, (\mathcal{F}_t)_{0\le t\le T})$,
  and $\{a_t,b_t\}_{0\le t\le T}$ be two $\F_t$-progressively measurable  processes such that
  $\int_0^T a_t^2 dt \ge \int_0^T b_t^2 dt$  $\P$-a.s. Define, for $0\le t\le T$, the  two processes
  $$Z^{a}_t :=  \exp \Big\{ -\frac 1 2 \int_0^t a^2_u du
  - \int_t^t a_u d B_u \Big\}, \quad
  Z^{b}_t :=  \exp \Big\{ -\frac 1 2 \int_0^t b^2_u du
  - \int_t^t b_u d B_u \Big\}.$$
  For any convex function $\psi:\mathbb{R} \to \mathbb{R}$, we have
  \begin{equation}
    \label{eq:conv}
    \mathbb{E}[ \psi (Z^{a}_T)] \ge
    \mathbb{E}[\psi(Z^{b}_T)].
  \end{equation}
\end{lemma}
\begin{proof}
  Define
  $$\tau^a(s) := \inf\{t\ge 0: \int_0^t a^2_u du>s\}, \qquad
  \tau^b(s) := \inf\{t\ge 0: \int_0^t b^2_u du>s\}.$$
  Then, since the  processes
  $\int_0^t a_u dB_u$ and
  $\int_0^t b_u dB_u$ are local martingales,
  the time-changed processes
  $$B_t^a := \int_0^{\tau^a(t)} a_u d B_u, \qquad
  B_t^b := \int_0^{\tau^b(t)} b_u d B_u$$
  are two standard Brownian motions  adapted to the
  time-changed filtrations
  $ \{\mathcal{F}_{\tau^a(t)}:t >0\}$
  and
  $\{\mathcal{F}_{\tau^b(t)}: t>0\}$ under the same
  probability space $(\Omega, \mathcal{F}, \mathbb{P})$,
  respectively.
  Define
  $$T^a := \int_0^T a_u^2 du, \qquad T^b := \int_0^T b_u^2 du.$$
  Then, it follows that $\tau^a(T^a) = \tau^b(T^b) = T$, and
  \begin{align}\mathbb{E}[\psi(Z_T^a)]
    = \mathbb{E}[  \psi(
    \exp\{ - \frac 1 2 T^a - B^a_{T^a}\})],
    \quad
    \mathbb{E}[\psi(Z_T^b)]
    = \mathbb{E}[  \psi(
    \exp\{ - \frac 1 2 T^b - B^b_{T^b}\})].\notag\end{align}
  With the martingale $\exp\{ - \frac 1 2 t - B^a_{t}\}$ and convex
  function   $\psi$,
  Jensen's inequality implies that $\psi(\exp\{ - \frac 1 2 t - B^a_{t}\})$
  is a submartingale. Also observe that
  $T^a \ge T^b$ almost surely in $\mathbb{P}$.
  Therefore,
  $\mathbb{E}[\psi(Z_T^a)] =  \mathbb{E}[
  \psi(
  \exp\{ - \frac 1 2 T^a - B^a_{T^a}\})] \ge
  \mathbb{E}[  \psi(
  \exp\{ - \frac 1 2 T^b - B^a_{T^b}\})] = \mathbb{E}[\psi(Z_T^b)].$
\end{proof}

Next, we proceed to prove Proposition \ref{prop-stochvol}.

\begin{proof} Applying Theorem \ref{thm:qhic},  the associated  randomized hypothesis testing is given by
  \[\widetilde{V}(x) = \inf_{a\ge 0, \lambda\in \Lambda} \{ xa +\mathbb{E}[ (1 - a \beta \tilde{Z}^{\lambda}_T S_T^\delta)^+]\},  \]
  where  according to \eqref{eq:sv} and \eqref{Ztil},
  \begin{align}\label{ZtilS} \tilde{Z}^\lambda_T S_T^\delta
    &= S_0^\delta \exp \Big\{ \delta (\delta -1)\int_0^T\frac{\sigma^2(Y_t)}{2}
    dt  \Big\}\cdot\exp \Big\{  -\int_{0}^T \frac{\lambda^2_t}{2} dt - \int_0^T \lambda_t d
    \hat{W}_t \Big\}\notag\\
    &\cdot \exp \Big\{ -\int_0^T\frac{(\delta \sigma(Y_t) -  \theta(Y_t))^2}{2}
    dt -  \int_0^T (  \theta(Y_t) - \delta \sigma(Y_t))d W_t \Big\} .
  \end{align}
  Note that for $\delta \in\{0,1\}$, $\tilde{Z}^\lambda S^\delta$ can be
  rewritten as
  $$\tilde{Z}^\lambda_T S_T^\delta
  = S_0^\delta
  \exp \Big\{ -\int_0^T\frac{\alpha_t^2  + \lambda_t^2}{2} dt - \int_0^T
  \sqrt{\alpha_t^2  + \lambda_t^2} d B_t \Big\},$$ where $\alpha_t :=
  \theta(Y_t) - \delta \sigma (Y_t)$ and $B_t$ is a standard Brownian motion
  defined by
\[d B_t = \frac{-\alpha_t d W_t - \lambda_t d \hat W_t}
  {\sqrt{\alpha_t^2 + \lambda_t^2}}.\]
  Hence $\tilde{Z}^\lambda S^\delta$ is in fact a $\P$-martingale
  for $\delta \in \{0,1\}$.
  In view of Lemma~\ref{prop:compz}, for any fixed $a\ge 0$, it is
  optimal to take $\hat \lambda \equiv 0$.
  Since  $\alpha_t^2$ is bounded positive process away from zero,
  applying Proposition~\ref{prop:tc} and Girsanov theorem,
  we have
  $\mathbb{P}\{ -\int_0^T \frac 1 2 \alpha_t^2 dt
  - \int_0^T \alpha_t dB_t = c \} = 0$, and hence
  $\mathbb{P}\{ \tilde{Z}^{\hat \lambda_T} S_T^\delta = c\} = 0$ holds
  for any constant $c$ and $\delta \in \{0,1\}$.
  To this end, we verified the second
  condition of Theorem~\ref{thm:qhic}, and conclude
  $\widetilde{V}(x) = V_1(x)= V(x)$ together with Proposition
  \ref{lem:npi}.
\end{proof}

Proposition  \ref{prop-stochvol} shows that  among all candidate EMMs
the MMM $\Q^0$  is optimal for $\widetilde{V}(x)$. In other words, when the benchmark is a constant or the stock $S_T$, the objective to  maximize the outperformance probability induces the investor to assign a zero  risk premium ($\lambda_t =0$)  for the second Brownian motion $\hat{W}$ under the stochastic factor model \eqref{eq:sv}-\eqref{eq:Y}.  Interestingly, this is true for all choices of $\theta$,  $\sigma$, $b$, $c$ and $\rho$ for $(S,Y)$. Furthermore, if  $\alpha_t =\theta(Y_t) - \delta \sigma (Y_t)$ is  constant, then the  expectation in \eqref{vstochvol} and hence the success probability $\widetilde{V}(x)$  can be computed explicitly.

\begin{corollary}  Suppose  $\theta(Y_t) - \delta \sigma (Y_t) =
  \alpha$  for some constant $\alpha\in \R\setminus \{0\}$.  Then, $\widetilde{V}(x)$ is given by
  \begin{equation}\label{vtil12}
    \begin{array}{ll}
      \widetilde V(x) & =      \left\{
        \begin{array}{ll}
          1, & \hbox{ if }  x\ge \beta S^\delta_0 \\
                    x \hat a + \Phi\big(d_2(\hat a;-\alpha)\big) - \hat aS^\delta_0 \Phi\big(d_1(\hat a;-\alpha)\big), & \hbox{ if } x< \beta S^\delta_0\\
        \end{array}\right.
    \end{array}
  \end{equation} where $d_1$ and $d_2$ are given in \eqref{dd} and  $\hat a$ in \eqref{eq:rhoa2}.
\end{corollary}

\subsubsection{General Benchmark and the HJB Characterization}
More generally, let us consider a stochastic benchmark in the form $F = f(S_T, Y_T)$ for some measurable function $f$.
The outperformance  portfolio optimization  is given by
\begin{equation}
  \label{eq:Vvoltsxy}
  \widetilde V(t, s, x, y)  = \sup_{\pi \in \mathcal{A}(x)}
  \mathbb{P}^{t,s,x,y} \{X^{x,\pi}_T \ge f(S_T, Y_T)\}, \end{equation}
where the notation $\mathbb{P}^{t,s,x,y}\{\cdot\} = \mathbb{P}\{\cdot \ |
S_t =s,  X_t = x, Y_t = y\}$.
We define:
\begin{equation}
  \label{eq:U2}
  U(t,s,y,z) :=  \inf_{\lambda\in \Lambda_t}
  \mathbb{E}^{t,s,y} [(1 -  Z^{z,\lambda}_T f(S_T,Y_T))^+],
\end{equation}
where $\mathbb{E}^{t,s,y}[\  \cdot \  ] = \mathbb{E}[\  \cdot\  | S_t = s,
Y_t = y]$, and $Z$ is given by
\begin{equation}
  \label{eq:Zl}
  Z^{z,\lambda}_u = z + \int_t^u Z^{z,\lambda}_\nu(-\theta (Y_\nu)
  d W_\nu - \lambda_\nu d \hat W_\nu).
\end{equation}
In view of Theorem \ref{thm:qhic}, if $\mathbb{P}\{Z_T^{a,\lambda}
f(S_T, Y_T) = 1\} = 0$ for all $a$, then we have
\begin{align}
  \widetilde V(t,s,x,y) &=
  \inf_{a\ge 0} \{ xa + \inf_{\lambda\in \Lambda_t}
  \mathbb{E}^{t,s,y} [(1 - a Z^{1,\lambda}_T f(S_T,Y_T))^+]\} \notag \\
  &=  \inf_{a\ge 0} \{ xa + \inf_{\lambda\in \Lambda_t}
  \mathbb{E}^{t,s,y} [(1 -  Z^{a,\lambda}_T f(S_T,Y_T))^+]\}\label{Vtsxy2}\\
  & =   \inf_{a\ge 0}  \{xa + U(t, s, y, a)\}.\label{Vtsxy3}
\end{align}

We shall derive the associated HJB PDE for $U$. To this end, we define, for
any scalar $\lambda\in \R$,  the differential operator
\begin{align*}
  \L^{\lambda} w &=    s \theta(y) \sigma(y)w_s +
  \frac 1 2  s^2 \sigma^2(y) w_{ss} +
  b(y)  w_y + \frac 1 2  c^2(y) w_{yy} \\ &+ \frac 1 2  (\theta^2(y) + \lambda^2) z^2 w_{zz}+
  s \sigma(y) c(y) \rho w_{sy}\\
  & -
  s z \sigma(y) \theta(y) w_{sz} +
  z c(y) (-\theta(y) \rho - \lambda \sqrt{1-\rho^2}) w_{yz}.
\end{align*}
Define the domains  $\mathcal{O} =  (0,\infty) \times (-\infty,\infty)
\times  (0,\infty)$,
$\mathcal{O}_T = (0,T) \times \mathcal{O}$. Also, denote by
$C^{2,1}(\mathcal{O}_T)$ the collection of all functions on
$\mathcal{O}_T$ which is continuously differentiable in $t$ and
continuously twice differentiable in $(s,y,z)$.

First, we have  the standard verification theorem, which is based on the
existence of a classical solution.
\begin{theorem}\label{thm:verification}
  If there exists $w\in C^{2,1}(\mathcal{O}_T) \cap
  C(\overline{\mathcal{O}}_T)$ satisfying the PDE:
  \begin{equation}
    \label{eq:hjb01}
    w_t +   \inf_{\lambda \in \mathbb{R}} \L^\lambda w = 0,
  \end{equation} with
  $w(T,s,y,z) = (1-zf(s,y))^+$, then    $w \le U$ holds on
  $\mathcal{O}_T$.
  Furthermore, if there exists a pair $(\hat Z, \hat  \lambda)$ of
  \eqref{eq:Zl},
  where $\hat \lambda$ is a feedback form of
  $\hat \lambda_\nu = \hat \lambda(\nu, S_\nu, Y_\nu, \hat Z_\nu)$
  satisfying
  \begin{equation}
    \label{eq:lambda}
    \L^{\hat \lambda(t,s,y,z)} w (t,s,y,z) = \inf_{\lambda\in
      \mathbb{R}} \L^\lambda w (t,s,y,z) = 0,
  \end{equation}
  then  $w = U$ holds on $\mathcal{O}_T$.

  Furthermore, if
  $\mathbb{P}\{Z_T^{a,\lambda} f(S_T, Y_T) = 1\} = 0$ for all $a$,
  then there exists
  $\hat a = \hat a(t,s,x,y)$ solves
  \begin{align}\label{hjbahat}
    \mathbb{E}^{t,s,y}[Z^{\hat a, \hat \lambda}_T f(S_T, Y_T)
    I_{\{Z^{\hat a, \hat \lambda}_T f(S_T, Y_T) <1 \}}] = \hat a x,
  \end{align}
  and
  \begin{align}
    \widetilde V(t,s,x,y)
    = \mathbb{P}^{t,s,x,y}\{  Z^{\hat a, \hat \lambda}_T f(S_T, Y_T)
    <1\}.
    \label{VeqP}
  \end{align}
\end{theorem}
\begin{proof}
  We follow the standard argument of verification theorem
  (Theorem 5.5.1 of \cite{YZ99}) in this below.
  First, for any $(S, Y, Z^\lambda)$ with initial $(s,y,z)$ at time
  $t$,
  we have
\begin{align*}w(t,s,y,z) + \mathbb{E}^{t,s,y,z}
  \Big[ \int_t^T \mathcal{L}^\lambda w(\nu, S_\nu, Y_\nu, Z_\nu) d \nu
  \Big]
  &=
  \mathbb{E}^{t,s,y,z}[ w (T, S_T, Y_T, Z_T)]\\
  &
  = \mathbb{E}^{t,s,y} [(1 -  Z^{z,\lambda}_T f(S_T,Y_T))^+].
\end{align*}
  The last equality above holds by terminal condition of
  PDE. Also observe that
  $ \mathbb{E}^{t,s,y,z}
  \Big[ \int_t^T \mathcal{L}^\lambda w(\nu, S_\nu, Y_\nu, Z_\nu) d \nu
  \Big]$ is always non-negative, and so we have
\[w(t,s,y,z) \le
  \mathbb{E}^{t,s,y} [(1 -  Z^{z,\lambda}_T f(S_T,Y_T))^+].
\]
  So, we conclude   $w\le U$ by arbitrariness of $\lambda$.
  On the other hand, if we take $\hat\lambda$  of \eqref{eq:lambda}
  in the above, then it yields equality, instead of inequality
\[w(t,s,y,z) =
  \mathbb{E}^{t,s,y} [(1 -  Z^{z,\lambda}_T f(S_T,Y_T))^+].\]
  By definition \eqref{eq:U2},
  we have right-hand side is always greater than or equal to $U$, and this
  implies $w\ge U$.

  Applying  \eqref{Vtsxy2}-\eqref{Vtsxy3}, the optimizer $\hat{a}$ for
  $V(t,s,x,y)$ is derived from \eqref{eq:2} of Theorem \ref{thm:rnp}
  with $\hat{H} = Z^{\hat a, \hat \lambda}_T f(S_T, Y_T)$ and $\hat{X}
  = I_{\{Z^{\hat a, \hat \lambda}_T f(S_T,Y_T) <1 \}}$.
  In turn, this yields \eqref{hjbahat} and
  \eqref{VeqP} via \eqref{eq:Vr}.
\end{proof}

Under quite general conditions, one can show that $U$ of \eqref{eq:U2} is
the unique solution of HJB equation \eqref{eq:hjb01}  in the viscosity
sense.
\begin{assumption}
  \label{a:2}
  $\theta(\cdot)$, $\mu(\cdot)$, $b(\cdot)$, $\sigma(\cdot)$,
  $f(\cdot,\cdot)$ and $c(\cdot)$ are all Lipschitz continuous.
\end{assumption}
\begin{proposition} \label{prop:vis}
  Under Assumption \ref{a:2}, the dual function $U$ in \eqref{eq:U2} is the unique
  bounded continuous viscosity solution of \eqref{eq:hjb01}
  with datum $w(T,s,y,z) = (1-zf(s,y))^+$ for all $(s,y,z) \in
  \mathcal{O}$.
\end{proposition}

\begin{proof} First, it can be shown that $U$ is the viscosity sub-solution (resp. supersolution)
  using the  Feynman-Kac formula on its super (resp. sub) test
  functions. For details, we refer to  the similar proof in  \cite[{Appendix}]{BSY11}.

  For uniqueness, we transform the domain from $\mathcal{O}$ to
  $\mathbb{R}$, by taking $x =(x_1, x_2, x_3):= (e^s, y, e^z)$ and
  defining $v(t,x) := w(t,s,y,z)$. Then,  \eqref{eq:hjb01} is equivalent to
  \begin{equation}
    \label{eq:hjb02}
    \inf_{\lambda \in \mathbb{R}} (v_t + \widetilde L^\lambda v)
    (t,x) = 0, \ (t,x) \in (0,T) \times \mathbb{R}^3,
  \end{equation}
  where
 \begin{align*}
    \widetilde L^\lambda v &=
    \displaystyle \frac 1 2 \sigma^2 (x_2) v_{x_1x_1} +
    \frac 1 2 c^2(x_2) v_{x_2x_2} +
    \frac 1 2 (\theta^2(x_2) + \lambda^2) v_{x_3x_3}   \\
    & +  \sigma(x_2) c(x_2) \rho v_{x_1x_2} -
    \sigma(x_2)  \theta(x_2) v_{x_1x_3} +
    c(x_2) (- \rho \theta(x_2) - \sqrt{1 - \rho^2} \lambda)
    v_{x_2x_3}   \\
    & + (\theta(x_2) - \frac 1 2 \sigma(x_2))\sigma(x_2)
    v_{x_1} +
    b(x_2) v_{x_2}
    - \frac 1 2 (\theta^2(x_2) + \lambda^2) v_{x_3}.
  \end{align*}
Now put in the standard form  \eqref{eq:hjb02}, the uniqueness of
  solution $v$, and thus $w$, follows from the comparison result  in
  \cite[Theorem 4.1]{GGIS91}.
\end{proof}

\section{Conclusions and Extensions}
\label{sec:conclusions}
We have studied the outperformance portfolio optimization problem in
complete and   incomplete markets.
The mathematical model is related to the generalized composite pure and randomized hypothesis testing
problems. We established  the connection between these two testing
problems and then used it  to address our portfolio optimization
problem. The maximal success probability exhibits special properties
with respect to benchmark scaling, while the outperformance portfolio
optimization does not enjoy economy of scale.  In various cases, we
obtained explicit solutions to the outperformance portfolio
optimization problem.  In the stochastic volatility model,  we showed
the special role played by the minimal martingale measure. With
the general benchmark, HJB characterization is available for
the outperformance probability. An alternative approach is the
characterization via BSDE solution for its dual representation (see \cite{MY99b, MZ02}).

There are a number of avenues for future research. Most naturally, one
can consider quantile hedging under other incomplete markets, with
specific  market frictions  and  trading constraints. Another extension involves claims with cash flows over different (random) times,
rather than a payoff  at a fixed terminal time, such as American options and insurance products.

On the other hand, the
composite nature of the hypothesis testing problems lends itself to
model uncertainty. To illustrate this point, let's consider a trader who receives $x$ from selling a contingent claim  with terminal random payoff $F \in [0,K]$ at time $T$. The objective is to minimize the risk of the terminal
liability $-F$ in terms of {\it Average Value at Risk}
\begin{align}\label{AVAR}AVaR(-F) &:= \max_{\mathbb{Q} \in \mathcal{Q}_\lambda}\mathbb{E}^\mathbb{Q}
[F]\\
\text{ subject to } &\inf_{Z\in \mathcal{Z}} \mathbb{E}[Z F] \ge x,\notag\end{align}
where the set of measures $\mathcal{Q}_\lambda := \{\mathbb{Q} \ll \mathbb{P} \,\Big|\,
\frac{d \mathbb{Q}}{d \mathbb{P}} \le \frac 1 \lambda,
\ \mathbb{P}-a.s.\}$ for  $\lambda \in (0,1]$.

In fact, we can convert this problem into a randomized composite hypothesis testing problem as in \eqref{eq:rnp}. To this end, we define $X := (K - F)/K$ and then  write  $AVaR(-F) = K - K V_\lambda(x)$,
where $V_\lambda(x)$ solves
\begin{align}V_\lambda(x) &= \sup_{X\in \mathcal{X}} \inf_{\mathbb{Q}\in
  \mathcal{Q}_\lambda} \mathbb{E}^{\mathbb{Q}}[X]\notag\\
  \text{ subject to } &\sup_{Z\in \mathcal{Z}} \mathbb{E}[ZX] \le \frac{K-x}{K}.\notag\end{align}
Following the  analysis in this paper, one can obtain the properties of the value function $V_\lambda(x)$ as well as the structure of the optimal solution.

Finally, the outperformance  portfolio optimization problem in Section \ref{sect-finance} is formulated with respect to a fixed reference
measure $\mathbb{P}$. This corresponds to applying
the theoretical results of Section 2 with the set $\mathcal{G} =\{1\}$; cf.
the proofs of Proposition~\ref{lem:npi} and Theorem~\ref{thm:qhic}.
It is also possible to incorporate  model uncertainty by replacing  the reference measure $\mathbb{P}$ by a class of probability measures
$\mathcal{M}$.  In this setup, the portfolio optimization problem becomes
 \begin{align*}
   {V}_{\mathcal{M}}(x)  := \sup_{\pi \in
    \mathcal{A}(x)} \inf_{\mathbb{M}\in\mathcal{M}}\mathbb{M}\{X^{x,\pi}_T \ge {F}\}, \qquad
  x\ge 0.
\end{align*}This is a special case of the hypothesis testing problems discussed in Section 2, where   the original set  $\mathcal{G}$ can be interpreted as the set containing the Radon-Nikodym densities $d\mathbb{M}/d\mathbb{P}$ with $\mathbb M \in \mathcal M$.
For related studies on the robust quantile hedging problem, we refer to \cite{Sch05, Sekine99}.

\subsection*{Acknowledgement}
The authors would like to thank two anonymous referees for their insightful remarks, as well as  Jun Sekine,  Birgit Rudloff and James
Martin for their helpful discussions.  Tim Leung's work is partially
supported by NSF grant DMS-0908295. Qingshuo Song's work is partially
supported by SRG grant 7002818 and GRF grant CityU 103310 of Hong Kong.

\appendix
\section{Appendix}\label{sec:app}
\subsection{The Role of $co(\mathcal{H})$ in $V(x)$}\label{app-a4}
In this example, we show that the representation of $V(x)$ in  \eqref{eq:Vr} does not hold  if $co(\mathcal{H})$ is replaced by the smaller set
$\mathcal{H}$.
\begin{example}\label{exm:coh}
  Let $\Omega = [0,1]$ and $\mathbb{P}$ be the Lebesgue measure,
  i.e. $\mathbb{P}(a,b) = b-a$ for $a\le b$. Let $\mathcal{G} = \{G
  \equiv 1\}$ and $\mathcal{H} = \{H_1, H_2\}$ with
  $$H_1(\omega) = I_{\{1/2 \le \omega \le 1\}} + 1, \quad
  H_2(\omega) = I_{\{0 \le \omega \le 1/2\}} + 1, \quad \omega\in \Omega.$$
  For the randomized hypothesis testing  problem \eqref{eq:rnp} with $x=1$, it is easy to see (e.g. from   \eqref{eq:Vr}) that
  \[ V(1) = \inf_{a \ge 0} \{xa + \inf_{\mathcal{G} \times co(\mathcal{H})}
  \mathbb{E} [(G- a H)^+]\} \Big|_{x=1} = \frac{2}{3},\]
  along with  the optimizers:
  $$\hat G  = 1, \quad \hat H = \frac 1 2 (H_1 + H_2), \quad \hat a = 2/3.$$
  In this simple example, uniqueness follows immediately.

  Now, if one switches from $co(\mathcal{H})$ to $\mathcal{H}$ in
  \eqref{eq:Vr},
  then  a strictly larger  value  will result:
  $$\inf_{a \ge 0} \{xa + \inf_{\mathcal{G} \times \mathcal{H}}
  \mathbb{E} [(G- a H)^+]\} \Big|_{x=1} = \frac{3}{4}>\frac{2}{3}=V(1).$$
\end{example}

\subsection{On the Positivity of $\hat{a}$}\label{app-a0}
First, we give  an example where the  minimizer $\hat{a}$ in
Theorem~\ref{thm:rnp} takes value zero, contrasting Proposition~3.1
and Lemma 4.3 in Cvitani{\'c} and Karatzas \cite{CK01}. Then, we provide a
sufficient condition for $\hat{a}>0$.

\begin{example}\label{exm:ce}
  Let $N_g := \bigcap_{G\in {\cal G}} \{G=0\}$ and $x>0$.
  \begin{itemize}
  \item[(i)] If $\mathbb{P}\{N_g\}=1$, then $\mathbb{E} [(G- aH)^+]
    \equiv 0$ for all $G, H, a$. Thus $\hat{a} = 0$ is the unique minimizer of  $\{xa +  \inf_{\mathcal{G} \times \mathcal{H}}
    \mathbb{E} [(G- a H)^+]\}$.
  \item[(ii)] If $0<\mathbb{P}\{N_g\} <1$ and $x > \sup_{H\in {\cal H}}
    \mathbb{E} [(H I_{N_g^c})]$,
    then there also exists a counter-example such that $\hat{a} = 0 $ minimizes  $\{xa +  \inf_{\mathcal{G}
      \times \mathcal{H}} \mathbb{E} [(G- a H)^+]\}$.     Indeed,  set ${\cal G}=\{G\}$ with
    $G={I_{N_g^c}}/{\mathbb{P}\{N_g^c\}}$ and
    ${\cal H} = \{H\}$ with $H \equiv 1$, then we have
    \begin{align}\label{eq:eg1}
      xa +  \inf_{\mathcal{G} \times \mathcal{H}}
      \mathbb{E} [(G- a H)^+] &= xa + \mathbb{E}[(G- zH)^+]\\ &=
      \left\{
        \begin{array}{cl}
          xa , & \text{if}\quad a \geq \frac{1}{P\{N_g^c\}};\\
          1+ a\left(x - \mathbb{P}\{N_g^c\}\right) , & \text{if}\quad  0\leq a <
          \frac{1}{\mathbb{P}\{N_g^c\}}.\notag
        \end{array}
      \right.
    \end{align}
    Since $x > \sup_{H\in {\cal H}}\limits \mathbb{E} [HI_{N_g^c}] =
    \mathbb{P}\{N_g^c\}$, $\hat{a} = 0$ is the unique minimizer of \eqref{eq:eg1}. $\Box$
  \end{itemize}
\end{example}

\begin{proposition}\label{prop:nz}
  If  \begin{equation}0<x<\sup_{\mathcal{H}}
    \mathbb{E} [ H I_{  \underset{G\in \mathcal{G}}{\cap}  \{G>0\}}],
    \label{asume1}\end{equation} then there exists
  $(\hat G, \hat H, \hat a, \hat X) \in
  \mathcal{G} \times \overline{co(\mathcal{H})}
  \times (0,\infty) \times
  \mathcal{X}_x$ satisfying \eqref{eq:1}-\eqref{eq:GX}. In particular,
  \[\hat a = \argmin_{a \geq 0} \{xa + \inf_{\mathcal{G}  \times
    \overline{co(\mathcal{H})}}  \mathbb{E} [(G- a
  H)^+]\}>0. \]
\end{proposition}
\begin{proof}
  Define the function   $f_x(a) := xa + \inf_{\mathcal{G}  \times
    \overline{co(\mathcal{H})}}  \mathbb{E} [(G- a
  H)^+]$, which is Lipschitz continuous (see Lemma 4.1 of \cite{CK01}).
  Since $f_x(0) = \inf_\mathcal{G} \mathbb{E}[G] \ge 0$ and is finite,
  and $\lim_{a\to \infty} f_x(a) = \infty$, there exists a finite
  $\hat a \ge 0$   that   minimizes $f_x(a)$.

  Now, suppose $\hat a = 0$ is a minimizer of
  $f_x(a)$. Then, it follows that $f_x(a) \ge f_x(0)$, $\forall a>0$, which leads to
  \begin{align} \label{eq:Gtil}
    \begin{array}{ll}
      xa &\ge \displaystyle  \inf_{\mathcal{G}} \mathbb{E}[G] -
      \inf_{\mathcal{G}  \times \overline{co(\mathcal{H})}}
      \mathbb{E} [(G- aH)^+]\\ \\ \displaystyle
      &  \ge \mathbb{E}[\tilde G] -
      \displaystyle \inf_{\overline{co(\mathcal{H})}} \mathbb{E}
      [(\tilde G- aH)^+]  \\ \\ &
      \ge \displaystyle
      a  \sup_{\overline{co(\mathcal{H})}} \mathbb{E} [H I_{\{\tilde G
        \ge aH\}}]
      \ge \displaystyle
      a  \sup_{\mathcal{H}} \mathbb{E} [H I_{\{\tilde G \ge aH\}}].
    \end{array}
  \end{align}
  In  \eqref{eq:Gtil}, $\tilde G$ minimizes  $\mathbb{E}[G]$ over
  $\mathcal{G}$, and its existence follows from convex and closedness
  of $\mathcal{G}$. Taking the limit $a \to 0^+$ yields a
  contradiction to \eqref{asume1}:
  $$x\ge \sup_{\mathcal{H}} \mathbb{E}[H I_{\{\tilde G >0\}}] \ge
  \sup_{\mathcal{H}} \mathbb{E} [ H I_{\displaystyle \cap_\mathcal{G} \{G>0\}}].$$
  Hence, we conclude that $\hat a>0$.
\end{proof}

\subsection{Counter-example  for Remark \ref{remark:concavemajor}}\label{sect-counter-major}
Let $\Omega=\{\omega_1, \omega_2\}$, $P(\{\omega_1\})=P(\{\omega_2\})=1/2$. Then, any random variable in ${\cal G}, {\cal H}$ or in $\mathcal{X}_x$, ${\cal I}_x$  can be represented as a point  in ${\mathbb R}^2$. Let ${\cal H}$ be line segment connecting $(2,4)$ and $(6,2)$, ${\cal G} = \{(2,2)\}$. Given $x\geq 0$, $\mathcal{X}_x$ is the convex quadrangle with four vertices
$(0,0), (x/3, 0), (x/5,2x/5)$, $(0,x/2)$ intersected
with $\{(x_1, x_2)\ |\ 0\leq x_1, x_2 \leq 1\}$. For each $H =(h_1,h_2)\in {\cal H}$ and $X=(x_1,x_2)$, the constraint  $\E[HX] \leq x$ implies that
$\frac{h_1}{2}x_1 + \frac{h_2}{2}x_2 \leq x.$
It is a half-plane bounded by $h_1x_1+h_2x_2 = 2x$, which passes $(x/5,2x/5)$ since $h_1+2h_2=5$. Hence, we have
\[V(x)=\sup_{(x_1,x_2)\in \mathcal{X}_x} x_1+x_2, \quad \text{ and } \quad V_1(x)=\sup_{(x_1,x_2)\in {\cal I}_x} x_1+x_2,\]
where ${\cal I}_x=\mathcal{X}_x\cap \{(0,0),(0,1),(1,0),(1,1)\}$. In  summary, the values are
\begin{center}
  \begin{tabular}{|c|c|c|}\hline
    $x$ & $V(x)$ & $V_1(x)$\\ \hline
    $0\leq x < 2$  &  $\frac{3}{5}x$ & $0$\\ \hline
    $2\leq x < \frac{5}{2}$ & $\frac{3}{5}x$ & $1$\\ \hline
    $\frac{5}{2}\leq x < 4$ & $\frac{x}{3}+\frac{2}{3}$ & $1$\\ \hline
    $x\geq 4$ & $2$ & $2$\\ \hline
  \end{tabular}
\end{center}
By inspecting the value of $V_1(x)$,  we see that its  smallest concave majorant must take value $\frac{x}{2}$ in $[0,4]$. Therefore,  $V(x)$ is not  the smallest concave majorant of  $V_1(x)$.

\subsection{Counter-example for Remark \ref{remark:C1}}\label{sec:dig}
With reference to Theorem \ref{thm:rnp},  we show via an example that one cannot remove the independence requirement when $\G$ and $H$ are not singletons.

\begin{example} \label{ex:dig}
  Let $\Omega = \{0,1\}\times [0,1]$, $\mathcal{F}_T =
  \mathcal{B}(\Omega)$. Let $\mu$ be the Lebesgue measure on $[0,1]$. Define
  $\mathbb{P}$ by
  $$\mathbb{P}(\{0\} \times A) = \mathbb{P}(\{1\} \times A) = \frac 1 2
  \mu(A), \quad  \forall A\in \mathcal{B}([0,1]).$$
  Let $H_0: \{0,1\} \to  \mathbb{R}$ be of $H_0(0) = 1/2$ and $H_0(1) =
  3/2$, and $f:[0,1]\to
  \mathbb{R}$ as an arbitrarily fixed probability density  function. Define the set
  $$\mathcal{H} = \{H: \Omega \to \mathbb{R}: H(\alpha, a) =
  H_0(\alpha) f(a), ~ (\alpha, a) \in \Omega \}$$
  and singleton $\mathcal{G} = \{G \equiv 1\}$. Let $U$ be a uniform random variable on $(\Omega, \mathcal{F}_T, \mathbb{P})$, such that $\mathbb{P}\{U\le a\} = a$ for $a\in [0,1]$.

  The  pure hypothesis testing problem is
  $$V_1 = \sup_{A\in \mathcal{F}_T} \mathbb{E}[I_A]$$
  subject to
  $$\sup_{H\in \mathcal{H}} \mathbb{E}[ HI_A] \le 1/2.$$
  Direct computation gives the     success set $\hat A = \{0\}$ and the
  value of pure hypothesis test $V_1 = 1/2$. On the other hand, the randomized hypothesis testing problem
  $$V = \sup_{X\in \mathcal{X}} \mathbb{E}[X]$$
  subject to
  $$\sup_{H\in \mathcal{H}} \mathbb{E}[ H X] \le 1/2.$$
  We find that $\hat H(\alpha, a) = H_0(\alpha)$
  and $\hat X = I_{ \{\alpha = 0\}} + 1/3 I_{\{\alpha = 1\}}$ solve
  this randomized hypothesis test with the optimal value $V = 2/3$.

  This shows that the values of pure  and randomized hypothesis tests are different. If one were to construct  an indicator version of the randomized test as in \eqref{Xhat}, namely,
  $$\bar X := I_{\{\alpha = 0\}} + I_{\{\alpha = 1\}} I_{\{U<1/3\}}.$$
  Although this test $\bar X $ still satisfies $\mathbb{E}[ \hat H \bar X] = 1/2$, it in fact does not solve either pure or
  randomized hypothesis test. Indeed,  for $\tilde H (\alpha, a) = H_0(\alpha) \cdot (3 I_{a<1/3}) \in \H$, we observe the violation: $\mathbb{E}[ \tilde H \bar X] = 1> 1/2$.\end{example}

\subsection{A property on non-degenerate martingale}
\label{sec:ndm}
On the probability space $(\Omega, \mathcal{F}, \mathbb{P})$ with filtration $(\mathcal{F}_t)_{0\le t\le 1}$, we denote by $W$ to be a standard Brownian motion. Let $Y$ be a $(\P, \F_t)$-martingale defined by
$$Y_t = \int_0^t \sigma_r d W_r, \quad t\in [0, 1].$$
where $\sigma_t$ is bounded $\mathcal{F}_t$-adapted process.

\begin{proposition}\label{prop:tc}
  Assume $c< \sigma_t<C$ for some positive constants $c$ and $C$,
  then
  $$\mathbb{P}\{Y_1 = b\} = 0$$
  for all constant $b$.
\end{proposition}
To prove this proposition, we will use following
two facts. We define
$f:  \mathbb{R}^+ \times \mathbb{R}^+
\times \mathbb{R} \mapsto [0,1]$ by
$$f(x,y,u) = \mathbb{P}\{W_t = u \hbox{ for some }
t\in (x,y)\}.$$
\begin{enumerate}
\item By direct computation, one can have
  $$\sup_{u\in \mathbb{R}} f(x,y,u) = f(x,y,0) < 1.$$
\item By a time-change argument, we have
  $$f(\lambda x, \lambda y, u) =
  f(x,y, \frac{u}{\sqrt \lambda}), \quad \forall \lambda>0.$$
\end{enumerate}

Now we are ready to present the proof
of  Proposition~\ref{prop:tc}.

\begin{proof}
  Since $Y$ is a continuous process,
  $$\{Y_1 = b\} \in  \sigma(\{\mathcal{F}_t: t<1\})=: \mathcal{F}_{1^-}.$$
  By Levy's zero one law, we have
  $$I_{\{Y_1 = b\}} = \lim_{t\uparrow 1}
  \mathbb{P}\{Y_1 = b |\mathcal{F}_t\}, \quad a.s.$$
  Therefore, it is enough to show that
  there exists $a\in (0,1)$ such that
  $$\mathbb{P}\{Y_1 = b|\mathcal{F}_t\} < a < 1,
 ~~ \forall t\in (0,1).$$
  Note that, the martingale $(Y_s|Y_t = u:s>t)$
  has the same
  distribution as a time-changed Brownian motion
  starting from state $u$. Together with
  $c^2 (1-t) \le \int_t^1 \sigma_r^2 dr
  \le C^2 (1-t)$, we have for some standard Brownian
  motion $B$ that
\begin{align}\mathbb{P}\{Y_1 = b | Y_t = u\} &=
  \mathbb{P}\{B_r = b - u,
  \hbox{ for some } r \in
  ( c^2(1-t), C^2(1-t)) \}\notag\\
   &=
  f(c^2, C^2, \frac{b - u}{\sqrt{1 - t}})
  \le f(c^2, C^2, 0).\notag
\end{align}
  Since $f(c^2, C^2, 0)$ is independent of $t$, and
  strictly less than $1$, we can
  simply take $a =  f(c^2, C^2, 0)$.
\end{proof}

To this end, one may wonder whether the condition on $\sigma$ in Proposition \ref{prop:tc} can
be relaxed to $\sigma_t>0$ a.s. $\forall t$. The answer is no, as shown by the counter-example in \cite{mof11}.

\bibliographystyle{amsplain}

\begin{small}

\end{small}
\end{document}